\documentclass[journal]{IEEEtran}\usepackage[stretch=50,shrink=50]{microtype}

\usepackage{epsfig,amsmath,amssymb,epsf,amsthm,scalefnt,multirow,bm,mathtools,stfloats,stmaryrd}
\usepackage{xcolor}
\usepackage{float}
\usepackage{cite}
\usepackage{psfrag}
\usepackage{algorithm}
\usepackage{algpseudocode}

\newtheorem{lemma}{Lemma}
\newtheorem*{lemma*}{Lemma}

\newtheorem{corollary}{Corollary}

\newcounter{mytempeqncnt}
\DeclareMathOperator*{\argmax}{argmax}

\algrenewcommand\algorithmicrequire{\textbf{Input:}}
\algrenewcommand\algorithmicensure{\textbf{Output:}}
\algdef{SE}[DOWHILE]{Do}{doWhile}{\algorithmicdo}[1]{\algorithmicwhile\ #1}

\begin{document}

\title{Spatial Wideband Channel Estimation for MmWave Massive MIMO Systems with Hybrid Architectures and Low-Resolution ADCs}

\author{In-soo Kim and Junil Choi
\thanks{The authors are with the School of Electrical Engineering, KAIST, Daejeon, Korea (e-mail: insookim@kaist.ac.kr; junil@kaist.ac.kr).}
\thanks{This work was partly supported by Institute of Information \& Communications Technology Planning \& Evaluation (IITP) grant funded by the Korea government (MSIT) (No. 2016-0-00123, Development of Integer-Forcing MIMO Transceivers for 5G \& Beyond Mobile Communication Systems), the MSIT (Ministry of Science and ICT), Korea, under the ITRC (Information Technology Research Center) support program (IITP-2020-0-01787) supervised by the IITP (Institute of Information \& Communications Technology Planning \& Evaluation), and the National Research Foundation (NRF) grant funded by the MSIT of the Korea government (2019R1C1C1003638).}}

\maketitle

\begin{abstract}
In this paper, a channel estimator for wideband millimeter wave (mmWave) massive multiple-input multiple-output (MIMO) systems with hybrid architectures and low-resolution analog-to-digital converters (ADCs) is proposed. To account for the propagation delay across the antenna array, which cannot be neglected in wideband mmWave massive MIMO systems, the discrete time channel that models the spatial wideband effect is developed. Also, the training signal design that addresses inter-frame, inter-user, and inter-symbol interferences is investigated when the spatial wideband effect is not negligible. To estimate the channel parameters over the continuum based on the maximum a posteriori (MAP) criterion, the Newtonized fully corrective forward greedy selection-cross validation-based (NFCFGS-CV-based) channel estimator is proposed. NFCFGS-CV is a gridless compressed sensing (CS) algorithm, whose termination condition is determined by the CV technique. The CV-based termination condition is proved to achieve the minimum squared error (SE). The simulation results show that NFCFGS-CV outperforms state-of-the-art on-grid CS-based channel estimators.
\end{abstract}

\section{Introduction}
Millimeter wave (mmWave) communications in the range of $30$-$300$ GHz and massive multiple-input multiple-output (MIMO) are promising technologies that offer high data rate and significant beamforming gain \cite{6894453}. The power consumption of mmWave massive MIMO systems, however, becomes impractical when a large number of radio frequency (RF) chains with unlimited hardware constraints are deployed. To alleviate the power consumption, one possible solution is to reduce the number of RF chains by using hybrid architectures \cite{7370753, 8030501}. In addition, using low-resolution analog-to-digital converters (ADCs) can simplify the hardware design, thereby reducing the power consumption \cite{7307134, 7894211, 7420605}.

In \cite{8171203, 8320852, 8310593, 8789658}, channel estimators were proposed for MIMO systems with low-resolution ADCs using on-grid compressed sensing (CS) algorithms. In particular, on-grid CS algorithms estimate continuous parameters (that are not guaranteed to be) on the grid, so these channel estimators suffer from the off-grid error, which occurs when the parameters of interest are off the grid. The generalized approximate message passing (GAMP) and vector AMP (VAMP) algorithms \cite{8171203} are belief propagation-based on-grid CS algorithms with low complexity, but these algorithms are sensitive to ill-conditioned sensing matrices. The generalized expectation consistent-signal recovery (GEC-SR) algorithm proposed in \cite{8320852} is a turbo principle-based on-grid CS algorithm. In \cite{8310593}, the channel estimation problem was formulated based on the sparse Bayesian learning (SBL) framework, which was solved using the variational Bayesian (VB) method. GEC-SR and VB-SBL are relatively robust to ill-conditioned sensing matrices, but matrix inversion is required, which may be infeasible due to the large matrix size in practice. The fully corrective forward greedy selection-cross validation (FCFGS-CV) algorithm proposed in \cite{8789658} is an orthogonal matching pursuit-based (OMP-based) on-grid CS algorithm, whose termination condition is determined by the CV technique. The relationship between CV and the optimal termination condition, however, was not analyzed in \cite{8789658}.

To combat the off-grid error, gridless CS algorithms were proposed for channel estimation in \cite{7491265, 8515242}. In essence, gridless CS algorithms estimate the parameters of interest on the infinite-resolution grid, thereby eliminating the off-grid error. The Newtonized OMP (NOMP) algorithm \cite{7491265} modifies the OMP algorithm \cite{4385788} by refining the estimated frequency over the continuum using Newton's method. NOMP, however, is applicable only to high-resolution ADCs. The atomic norm minimization-based channel estimator was proposed in \cite{8515242} for mixed one-bit ADCs, but tuning the hyperparameter of the atomic norm minimization requires high-resolution ADCs.

The prior work on channel estimation for MIMO systems with low-resolution ADCs, however, did not consider the spatial wideband effect. In wideband mmWave massive MIMO systems, the propagation delay across the antenna array cannot be neglected due to the large array and small sampling period, which is a phenomenon called the spatial wideband effect \cite{8354789}. In \cite{8354789, 8647494}, the beam squint effect, which is the frequency domain manifestation of the spatial wideband effect, was mainly analyzed because channel estimation was performed in orthogonal frequency division multiplexing (OFDM) systems. For low-resolution ADCs, however, single-carrier systems are mainly considered because the orthogonality of OFDM systems cannot be preserved \cite{7472305}, implying that the spatial wideband effect should be analyzed in the time domain. The spatial wideband effect in the time domain was considered in the context of RF signal processing \cite{1069}, but the impact of the spatial wideband effect on wireless communications, such as the new channel model or the requirements on the training signal design, has not been studied before.

In this paper, the Newtonized FCFGS-CV-based (NFCFGS-CV-based) channel estimator for wideband mmWave massive MIMO systems with hybrid architectures and low-resolution ADCs is proposed. To account for the propagation delay across the antenna array, the discrete time channel is formulated based on the spatial wideband effect. In addition, the training signal design that mitigates inter-frame, inter-user, and inter-symbol interferences is addressed when the spatial wideband effect is not negligible. Then, the channel estimation problem is developed using the maximum a posteriori (MAP) criterion. To estimate the parameters of the channel over the continuum, the NFCFGS algorithm is proposed, which is a new gridless CS algorithm that combines NOMP \cite{7491265} and FCFGS \cite{doi:10.1137/090759574} that were originally developed separately. NOMP is a gridless CS algorithm for high-resolution ADCs, whereas FCFGS is an on-grid CS algorithm for low-resolution ADCs. The proposed NFCFGS, therefore, can be interpreted as gridless FCFGS, whose gridless property is adopted from NOMP. To determine the termination condition that achieves the minimum squared error (SE), the CV technique \cite{4301267} is applied. In particular, the quality of the estimate is assessed based on the CV function, which is the log-likelihood function of the CV data. The CV data is a portion of the received signal that is excluded from channel estimation, and reserved solely for the assessment of the estimation quality. The analysis on CV shows that the CV function is proportional to the SE, so using CV as an indicator of termination is justified. Then, the performance of NFCFGS-CV is evaluated in the simulation results.

This paper is organized as follows. The discrete time channel with the spatial wideband effect is formulated in Section \ref{section_2}. In Section \ref{section_3}, the system model with hybrid architectures and low-resolution ADCs is described to illustrate how the training signals are transmitted and received. The NFCFGS-CV-based channel estimator is proposed in Section \ref{section_4}, whose CV technique is analyzed in the asymptotic regime when the number of the CV data is sufficiently large. The performance of NFCFGS-CV is evaluated based on the simulation results in Section \ref{section_5}. The conclusion follows in Section \ref{section_6}.

\textbf{Notation:} $a$, $\mathbf{a}$, and $\mathbf{A}$ denote a scalar, vector, and matrix. The $i$-th element of $\mathbf{a}$ is $a_{i}$, and the transpose of the $i$-th row of $\mathbf{A}$ is $\mathbf{a}_{i}$. The row restriction of $\mathbf{A}$ to the index set $\mathcal{I}$ is $\mathbf{A}_{\mathcal{I}}$. The $2$-norm of $\mathbf{a}$ is $\|\mathbf{a}\|$. The vectorization of $\mathbf{A}$ is $\mathrm{vec}(\mathbf{A})$. The Kronecker product of $\mathbf{A}$ and $\mathbf{B}$ is $\mathbf{A}\otimes\mathbf{B}$. $\mathbf{A}\succ\mathbf{B}$ implies that $\mathbf{A}-\mathbf{B}$ is positive definite. $\llbracket n\rrbracket$ denotes $\llbracket n\rrbracket=\{1, \dots, n\}$.

\section{Channel with Spatial Wideband Effect}\label{section_2}
Assume a single-cell uplink massive MIMO system with a base station and $K$ single-antenna users. The base station is equipped with a uniform linear array (ULA) of $M$ antennas. The system operates under the single-carrier transmission in the mmWave wideband with the carrier frequency $f_{c}$, carrier wavelength $\lambda_{c}$, bandwidth $W$, and sampling period $T_{s}$, which means that $\lambda_{c}=c/f_{c}$ and $T_{s}=1/W$ where $c$ is the speed of light. In this section, the channel that accounts for the spatial wideband effect is formulated.

Assume that the channel is formed by $L_{k}$ paths between the base station and $k$-th user where the $\ell$-th path is associated with the path gain $\bar{\alpha}_{k, \ell}\in\mathbb{C}$, angle-of-arrival (AoA) $\theta_{k, \ell}\in[-\pi/2, \pi/2]$, and delay $\tau_{k, \ell}\in[0, (D-1)T_{s}]$ where $D\in\mathbb{N}$ defines the delay spread. Then, the propagation delay of the $\ell$-th path from the $k$-th user to the $m$-th antenna is
\begin{equation}\label{propagation_delay}
\tau_{k, \ell, m}=\tau_{k, \ell}+(m-1)\frac{\mathsf{d}\sin(\theta_{k, \ell})}{c}
\end{equation}
as shown in Fig. \ref{figure_1} where $\mathsf{d}$ is the antenna spacing. Therefore, the received signal at the $m$-th antenna from the $k$-th user at time $t$ is
\begin{align}
r_{k, m}(t)&=\sum_{\ell=1}^{L_{k}}\bar{\alpha}_{k, \ell}e^{-j2\pi f_{c}\tau_{k, \ell, m}}s_{k}(t-\tau_{k, \ell, m})\notag\\
           &=\sum_{\ell=1}^{L_{k}}\alpha_{k, \ell}e^{-j2\pi f_{c}(m-1)\frac{\mathsf{d}\sin(\theta_{k, \ell})}{c}}s_{k}(t-\tau_{k, \ell, m})
\end{align}
where $\alpha_{k, \ell}=\bar{\alpha}_{k, \ell}e^{-j2\pi f_{c}\tau_{k, \ell}}$ is the equivalent path gain, $s_{k}(t)$ is the training signal of the $k$-th user at time $t$ defined as
\begin{equation}
s_{k}(t)=\sum_{i=-\infty}^{\infty}s_{k}[i]p(t-iT_{s}),
\end{equation}
$s_{k}[i]\in\mathbb{C}$ represents the $i$-th training symbol of the $k$-th user with a transmit power constraint $\mathbb{E}\{|s_{k}[i]|^{2}\}\leq\rho_{k}$, and $p(t)$ is the pulse shaping filter. Then, the received signal at the $m$-th antenna from the $K$ users at time $t$ is
\begin{equation}\label{received_signal_with_noise}
r_{m}(t)=\sum_{k=1}^{K}r_{k, m}(t)+\bar{v}_{m}(t)
\end{equation}
where $\bar{v}_{m}(t)$ is the zero-mean additive white Gaussian noise (AWGN) at the $m$-th antenna at time $t$. As a shorthand notation, denote the total number of paths as $L=L_{1}+\cdots+L_{K}$.

\begin{figure}[t]
\centering
\includegraphics[width=1\columnwidth]{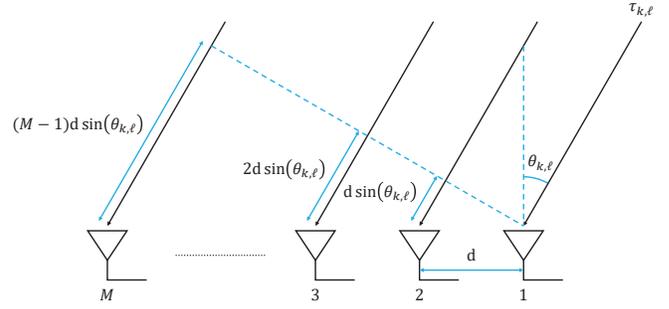}
\caption[caption]{The propagation delay across the ULA with $M$ antennas, which results from the large array and small sampling period, incurs a phenomenon called the spatial wideband effect.}\label{figure_1}
\end{figure}

The fact that should be emphasized is that in narrowband small-scale MIMO systems, the propagation delay across the antenna array was omitted, namely $s_{k}(t-\tau_{k, \ell, m})\approx s_{k}(t-\tau_{k, \ell})$, due to the small array and large sampling period. In wideband mmWave massive MIMO systems where the array is large and sampling period is small, however, the propagation delay across the antenna array cannot be neglected, which is a phenomenon called the spatial wideband effect \cite{8354789}. The spatial wideband effect, which was once overlooked in wireless communications, was recently studied extensively in \cite{8354789, 8647494} due to the emergence of mmWave communications.

Now, the discrete time channel that accounts for the spatial wideband effect is formulated, which was not considered in the prior work to the best of our knowledge. Consider the sampled received signal at the $m$-th antenna from the $k$-th user at time $nT_{s}$, which is
\begin{align}\label{received_signal}
r_{k, m}[n]\overset{\hphantom{(a)}}{=}&r_{k, m}(nT_{s})\notag\\
           \overset{\hphantom{(a)}}{=}&\sum_{\ell=1}^{L_{k}}\sum_{i=-\infty}^{\infty}\alpha_{k, \ell}e^{-j2\pi f_{c}(m-1)\frac{\mathsf{d}\sin(\theta_{k, \ell})}{c}}\times\notag\\
                                      &p((n-i)T_{s}-\tau_{k, \ell, m})s_{k}[i]\notag\\
                      \overset{(a)}{=}&\sum_{\ell=1}^{L_{k}}\sum_{d=-\infty}^{\infty}\alpha_{k, \ell}e^{-j2\pi f_{c}(m-1)\frac{\mathsf{d}\sin(\theta_{k, \ell})}{c}}\times\notag\\
                                      &p(dT_{s}-\tau_{k, \ell, m})s_{k}[n-d]
\end{align}
where (a) is due to the change of variables $i=n-d$. In practice, $p(t)$ decays rapidly outside the mainlobe, whose interval is typically defined as $\{t\in\mathbb{R}\mid t\in(-T_{s}, T_{s})\}$. Therefore, \eqref{received_signal} can be approximated by retaining only the terms corresponding to $\{d\in\mathbb{Z}\mid dT_{s}-\tau_{k, \ell, m}\in(-T_{s}, T_{s})\}$ \cite{tse2005fundamentals}. From the fact that $\tau_{k, \ell, m}\in[-(M-1)\mathsf{d}/c, (D-1)T_{s}+(M-1)\mathsf{d}/c]$, which follows from \eqref{propagation_delay} and $\tau_{k, \ell}\in[0, (D-1)T_{s}]$, we conclude that the terms corresponding to
\begin{equation}\label{channel_tap}
\mathcal{D}=\left\{d\in\mathbb{Z}\mid d\in\left(-1-\frac{(M-1)\frac{\mathsf{d}}{c}}{T_{s}}, D+\frac{(M-1)\frac{\mathsf{d}}{c}}{T_{s}}\right)\right\}
\end{equation}
are sufficient to approximate \eqref{received_signal}, which gives
\begin{align}\label{approximation}
r_{k, m}[n]\approx&\sum_{d\in\mathcal{D}}\sum_{\ell=1}^{L_{k}}\alpha_{k, \ell}e^{-j2\pi f_{c}(m-1)\frac{\mathsf{d}\sin(\theta_{k, \ell})}{c}}\times\notag\\
                  &p(dT_{s}-\tau_{k, \ell, m})s_{k}[n-d].
\end{align}
Before moving on, define
\begin{equation}\label{loup}
\begin{aligned}
&D_{\mathrm{lo}}=\min\mathcal{D}=-\lceil(M-1)\mathsf{d}/(cT_{s})\rceil,\\
&D_{\mathrm{up}}=\max\mathcal{D}=D-1+\lceil(M-1)\mathsf{d}/(cT_{s})\rceil
\end{aligned}
\end{equation}
for notational simplicity.

From \eqref{approximation}, the $d$-th channel tap of the $k$-th user over the $M$ antennas is obtained as
\begin{align}\label{spatial_wideband_effect}
\mathbf{h}_{k}[d]&=\sum_{\ell=1}^{L_{k}}\alpha_{k, \ell}\underbrace{\begin{bmatrix}p(dT_{s}-\tau_{k, \ell, 1})\\\vdots\\e^{-j2\pi f_{c}(M-1)\frac{\mathsf{d}\sin(\theta_{k, \ell})}{c}}p(dT_{s}-\tau_{k, \ell, M})\end{bmatrix}}_{=\mathbf{a}_{d}(\theta_{k, \ell}, \tau_{k, \ell})\in\mathbb{C}^{M}}\notag\\
                 &=\underbrace{\begin{bmatrix}\mathbf{a}_{d}(\theta_{k, 1}, \tau_{k, 1})&\cdots&\mathbf{a}_{d}(\theta_{k, L_{k}}, \tau_{k, L_{k}})\end{bmatrix}}_{=\mathbf{F}_{k}[d]\in\mathbb{C}^{M\times L_{k}}}\underbrace{\begin{bmatrix}\alpha_{k, 1}\\\vdots\\\alpha_{k, L_{k}}\end{bmatrix}}_{=\bm{\alpha}_{k}\in\mathbb{C}^{L_{k}}}.
\end{align}
Then, the sampled received signal over the $M$ antennas from the $K$ users at time $nT_{s}$ is given by \eqref{received_signal_with_noise}, \eqref{approximation}, and \eqref{spatial_wideband_effect} as
\begin{align}\label{H}
\mathbf{r}[n]&=\sum_{d\in\mathcal{D}}\underbrace{\begin{bmatrix}\mathbf{h}_{1}[d]&\cdots&\mathbf{h}_{K}[d]\end{bmatrix}}_{=\mathbf{H}[d]\in\mathbb{C}^{M\times K}}\underbrace{\begin{bmatrix}s_{1}[n-d]\\\vdots\\s_{K}[n-d]\end{bmatrix}}_{=\mathbf{s}[n-d]\in\mathbb{C}^{K}}+\bar{\mathbf{v}}[n]\notag\\
             &=\underbrace{\begin{bmatrix}\mathbf{H}[D_{\mathrm{lo}}]&\cdots&\mathbf{H}[D_{\mathrm{up}}]\end{bmatrix}}_{=\mathbf{H}\in\mathbb{C}^{M\times |\mathcal{D}|K}}\underbrace{\begin{bmatrix}\mathbf{s}[n-D_{\mathrm{lo}}]\\\vdots\\\mathbf{s}[n-D_{\mathrm{up}}]\end{bmatrix}}_{=\mathbf{s_{n}}\in\mathbb{C}^{|\mathcal{D}|K}}+\bar{\mathbf{v}}[n]
\end{align}
where $\bar{\mathbf{v}}[n]\sim\mathcal{CN}(\mathbf{0}_{M}, \mathbf{I}_{M})$ is the sampled AWGN over the $M$ antennas at time $nT_{s}$, and the signal-to-noise ratio (SNR) is defined as $1/K\cdot(\rho_{1}+\cdots+\rho_{K})$. From now on, the discrete time channel with the spatial wideband effect, namely $\mathbf{H}$, is considered throughout the paper. The parameters of $\mathbf{H}$ are assumed to be independent and identically distributed (i.i.d.) as
\begin{equation}\label{path_gain_aoa_delay}
\begin{aligned}
\alpha_{k, \ell}&\sim\mathcal{CN}(0, 1),\\
\theta_{k, \ell}&\sim\mathrm{Uniform}([-\pi/2, \pi/2]),\\
  \tau_{k, \ell}&\sim\mathrm{Uniform}([0, (D-1)T_{s}])
\end{aligned}
\end{equation}
for all $(k, \ell)$.

\textbf{Remark 1:} The spatial wideband channel has two distinct properties. First, there are $2\lceil(M-1)\mathsf{d}/(cT_{s})\rceil$ more channel taps compared to the number of channel taps that the narrowband channel has, namely $D$. This can be checked from \eqref{channel_tap} and \eqref{loup}. Second, each antenna receives different samples of $p(t)$, which is evident from \eqref{spatial_wideband_effect}. In narrowband small-scale MIMO systems where the spatial wideband effect is negligible due to the small array and large sampling period, the spatial wideband channel reduces to the conventional channel \cite{7400949}, which can be verified from the fact that $(M-1)\mathsf{d}/(cT_{s})\approx 0$ and $p(dT_{s}-\tau_{k, \ell, m})\approx p(dT_{s}-\tau_{k, \ell})$.

\section{System Model}\label{section_3}
In this section, the system model that describes how the training signals are transmitted and received is formulated based on the discussions in the previous section. The base station employs a hybrid architecture with $R$ RF chains where $R\leq M$. To reduce the power consumption at the base station, each RF chain is equipped with a pair of $B$-bit ADCs as illustrated in Fig. \ref{figure_2}.

\begin{figure}[t]
\centering
\includegraphics[width=1\columnwidth]{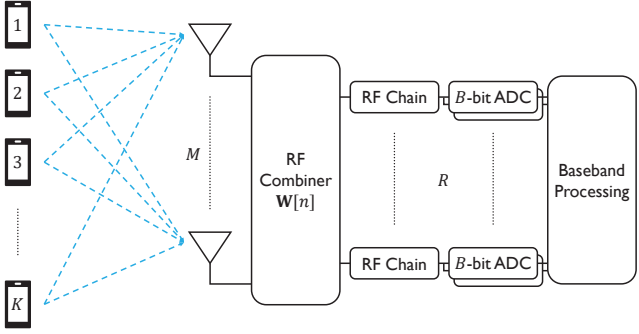}
\caption[caption]{A single-cell uplink between a base station with a hybrid architecture and $K$ single-antenna users. The base station is equipped with a ULA of $M$ antennas and $R$ RF chains where a pair of $B$-bit ADCs are deployed at each RF chain.}\label{figure_2}
\end{figure}

The RF-combined received signal $\mathbf{y}[n]\in\mathbb{C}^{R}$ at time $n$ is
\begin{align}\label{h}
\mathbf{y}[n]&\overset{\hphantom{(a)}}{=}\mathbf{W}[n]^{\mathrm{H}}\mathbf{r}[n]\notag\\
             &\overset{\hphantom{(a)}}{=}\mathbf{W}[n]^{\mathrm{H}}\mathbf{H}\mathbf{s}_{n}+\mathbf{W}[n]^{\mathrm{H}}\bar{\mathbf{v}}[n]\notag\\
             &\overset{(a)}{=}(\underbrace{\mathbf{s}_{n}^{\mathrm{T}}\otimes\mathbf{W}[n]^{\mathrm{H}}}_{=\bar{\mathbf{A}}[n]})\underbrace{\mathrm{vec}(\mathbf{H})}_{=\mathbf{h}}+\underbrace{\mathbf{W}[n]^{\mathrm{H}}\bar{\mathbf{v}}[n]}_{=\mathbf{v}[n]}
\end{align}
where (a) follows from $\mathrm{vec}(\mathbf{A}\mathbf{B}\mathbf{C})=(\mathbf{C}^{\mathrm{T}}\otimes\mathbf{A})\mathrm{vec}(\mathbf{B})$, and $\mathbf{W}[n]\in\mathbb{C}^{M\times R}$ is the RF combiner at time $n$ with a network of phase shifters. The phase shifters are constrained to satisfy $\mathbf{W}[n]^{\mathrm{H}}\mathbf{W}[n]=\mathbf{I}_{R}$, so $\mathbf{v}[n]\sim\mathcal{CN}(\mathbf{0}_{R}, \mathbf{I}_{R})$.

The training signals are transmitted in $N_{\mathrm{t}}$ frames where each frame consists of $N_{\mathrm{f}}$ symbols. Since the phase shifters cannot be reconfigured at each $n$, $\mathbf{W}[n]$ is fixed at each frame \cite{7961152}. In addition, prefix guard intervals of length $N_{\mathrm{p}}$ symbols and suffix guard intervals of length $N_{\mathrm{s}}$ symbols are appended to each frame to avoid inter-frame interference. In contrast, no suffix guard intervals are required in the conventional channel model. The details of the frame, prefix, and suffix design along with the training signal design that mitigates inter-frame, inter-user, and inter-symbol interferences in the presence of the spatial wideband effect are discussed in Remarks 2, 3, and 4. For notational convenience, denote time $n$ that corresponds to the $n_{\mathrm{f}}$-th symbol of the $n_{\mathrm{t}}$-th frame as $(n_{\mathrm{t}}, n_{\mathrm{f}})$. For example, $\mathbf{y}[n]$ and $\mathbf{s}_{n}$ that correspond to the $n_{\mathrm{f}}$-th symbol of the $n_{\mathrm{t}}$-th frame are written as $\mathbf{y}[(n_{\mathrm{t}}, n_{\mathrm{f}})]$ and $\mathbf{s}_{(n_{\mathrm{t}}, n_{\mathrm{f}})}$. An overview of the frame structure is shown in Fig. \ref{figure_3}. Before moving on, denote the length of the channel estimation phase as $N=N_{\mathrm{t}}N_{\mathrm{f}}$.

The RF-combined received signal $\mathbf{y}$ over the channel estimation phase of length $N$ is
\begin{equation}\label{unquantized_received_signal}
\mathbf{y}=\begin{bmatrix}\mathbf{y}[(1, 1)]\\\vdots\\\mathbf{y}[(N_{\mathrm{t}}, N_{\mathrm{f}})]\end{bmatrix}=\underbrace{\begin{bmatrix}\bar{\mathbf{A}}[(1, 1)]\\\vdots\\\bar{\mathbf{A}}[(N_{\mathrm{t}}, N_{\mathrm{f}})]\end{bmatrix}}_{=\bar{\mathbf{A}}\in\mathbb{C}^{RN\times M|\mathcal{D}|K}}\mathbf{h}+\underbrace{\begin{bmatrix}\mathbf{v}[(1, 1)]\\\vdots\\\mathbf{v}[(N_{\mathrm{t}}, N_{\mathrm{f}})]\end{bmatrix}}_{=\mathbf{v}\in\mathbb{C}^{RN}}.
\end{equation}
Then, $\mathbf{y}$ is quantized by a pair of $B$-bit ADCs at each RF chain. The quantized received signal $\hat{\mathbf{y}}$ over the channel estimation phase of length $N$ is
\begin{equation}
\hat{\mathbf{y}}=\mathrm{Q}(\mathbf{y})
\end{equation}
where $\mathrm{Q}(\cdot)$ is the $B$-bit quantization function applied elementwise as
\begin{equation}\label{quantization}
\hat{y}=\mathrm{Q}(y)\iff\begin{cases}\mathrm{Re}(y)\in[\mathrm{Re}(\hat{y}^{\mathrm{lo}}), \mathrm{Re}(\hat{y}^{\mathrm{up}}))\\\mathrm{Im}(y)\in[\mathrm{Im}(\hat{y}^{\mathrm{lo}}), \mathrm{Im}(\hat{y}^{\mathrm{up}}))\end{cases},
\end{equation}
and $\hat{y}^{\mathrm{lo}}\in\mathbb{C}$ and $\hat{y}^{\mathrm{up}}\in\mathbb{C}$ are the lower and upper thresholds associated with $\hat{y}\in\mathbb{C}$. The real and imaginary parts of $\hat{y}^{\mathrm{lo}}$, $\hat{y}^{\mathrm{up}}$, and $\hat{y}$ correspond to one of the $2^{B}$ quantization intervals. For notational simplicity, denote $\mathcal{Q}$ as the collection of the $2^{B}$ quantization points, namely $\mathrm{Re}(\hat{y})\in\mathcal{Q}$ and $\mathrm{Im}(\hat{y})\in\mathcal{Q}$.

To estimate $\mathbf{h}$ from $\hat{\mathbf{y}}$, express $\mathbf{h}$ as
\begin{equation}
\mathbf{h}=\underbrace{\begin{bmatrix}\mathbf{F}_{1}[D_{\mathrm{lo}}]&&\\&\ddots&\\&&\mathbf{F}_{K}[D_{\mathrm{lo}}]\\&\vdots&\\\mathbf{F}_{1}[D_{\mathrm{up}}]&&\\&\ddots&\\&&\mathbf{F}_{K}[D_{\mathrm{up}}]\end{bmatrix}}_{=\mathbf{F}(\mathcal{P})\in\mathbb{C}^{M|\mathcal{D}|K\times L}}\underbrace{\begin{bmatrix}\bm{\alpha}_{1}\\\vdots\\\bm{\alpha}_{K}\end{bmatrix}}_{=\bm{\alpha}\in\mathbb{C}^{L}},
\end{equation}
which follows from the definition of $\mathbf{h}$ in \eqref{spatial_wideband_effect}, \eqref{H}, and \eqref{h}, and $\mathcal{P}$ is the collection of all $(\theta_{k, \ell}, \tau_{k, \ell})$. Then, \eqref{unquantized_received_signal} becomes
\begin{equation}\label{modified_unquantized_received_signal}
\mathbf{y}=\underbrace{\bar{\mathbf{A}}\mathbf{F}(\mathcal{P})}_{=\mathbf{A}(\mathcal{P})}\bm{\alpha}+\mathbf{v}
\end{equation}
where $\mathbf{A}(\mathcal{P})\in\mathbb{C}^{RN\times L}$.

Now, the goal is to estimate $(\bm{\alpha}, \mathcal{P})$ from $\hat{\mathbf{y}}$. Using the MAP criterion, the path gain estimation problem can be formulated as a convex optimization problem for fixed $\mathcal{P}$ \cite{7439790}. The channel estimation problem, however, is highly nonconvex because $\mathbf{F}(\mathcal{P})$ is nonlinear with respect to $\mathcal{P}$. Furthermore, the number of paths is unknown in practice, which complicates the situation. To deal with the channel estimation problem, the NFCFGS-CV-based channel estimator is developed in Section \ref{section_4}.

\begin{figure}[t]
\centering
\includegraphics[width=1\columnwidth]{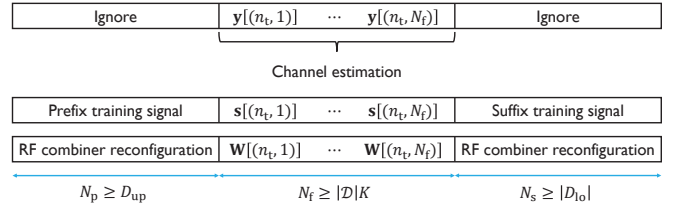}
\caption[caption]{The structure of a frame with prefix and suffix guard intervals. The RF combiner is fixed at each frame, namely $\mathbf{W}[(n_{\mathrm{t}}, 1)]=\cdots=\mathbf{W}[(n_{\mathrm{t}}, N_{\mathrm{f}})]$. In the presence of the spatial wideband effect, $\mathbf{y}[(n_{\mathrm{t}}, n_{\mathrm{f}})]$ receives delayed copies of $\mathbf{s}[(n_{\mathrm{t}}, n_{\mathrm{f}}-D_{\mathrm{lo}})]$, \dots, $\mathbf{s}[(n_{\mathrm{t}}, n_{\mathrm{f}}-D_{\mathrm{up}})]$, so the requirements on $N_{\mathrm{p}}$ and $N_{\mathrm{s}}$ to avoid inter-frame interference become $N_{\mathrm{p}}\geq D_{\mathrm{up}}\geq D-1$ and $N_{\mathrm{s}}\geq|D_{\mathrm{lo}}|\geq 0$. In the narrowband channel, on the other hand, $D_{\mathrm{lo}}$ and $D_{\mathrm{up}}$ reduces to $0$ and $(D-1)$, so no suffix guard intervals are required.}\label{figure_3}
\end{figure}

\textbf{Remark 2:} In the conventional channel model, prefix guard intervals of length $(D-1)$ are sufficient to avoid inter-frame interference. To avoid inter-frame interference in the presence of the spatial wideband effect, however, the requirements on $N_{\mathrm{p}}$ and $N_{\mathrm{s}}$ become $N_{\mathrm{p}}\geq D_{\mathrm{up}}$ and $N_{\mathrm{s}}\geq|D_{\mathrm{lo}}|$, as evident from \eqref{channel_tap} and \eqref{loup}. The overhead of the guard intervals is small when $N_{\mathrm{f}}\gg\max\{N_{\mathrm{p}}, N_{\mathrm{s}}\}$.

\textbf{Remark 3:} Consider the RF-combined received signal at the $n_{\mathrm{t}}$-th frame, which is
\begin{align}
\begin{bmatrix}\mathbf{y}[(n_{\mathrm{t}}, 1)]\\\vdots\\\mathbf{y}[(n_{\mathrm{t}}, N_{\mathrm{f}})]\end{bmatrix}=&\left(\vphantom{\begin{bmatrix}\mathbf{s}_{(n_{\mathrm{t}}, 1)}^{\mathrm{T}}\\\vdots\\\mathbf{s}_{(n_{\mathrm{t}}, N_{\mathrm{f}})}^{\mathrm{T}}\end{bmatrix}}\right.\underbrace{\begin{bmatrix}\mathbf{s}_{(n_{\mathrm{t}}, 1)}^{\mathrm{T}}\\\vdots\\\mathbf{s}_{(n_{\mathrm{t}}, N_{\mathrm{f}})}^{\mathrm{T}}\end{bmatrix}}_{=\mathbf{S}}\otimes\mathbf{W}[(n_{\mathrm{t}}, 1)]^{\mathrm{H}}\left.\vphantom{\begin{bmatrix}\mathbf{s}_{(n_{\mathrm{t}}, 1)}^{\mathrm{T}}\\\vdots\\\mathbf{s}_{(n_{\mathrm{t}}, N_{\mathrm{f}})}^{\mathrm{T}}\end{bmatrix}}\right)\mathbf{h}\notag\\
                                                                                                                  &+\begin{bmatrix}\mathbf{v}[(n_{\mathrm{t}}, 1)]\\\vdots\\\mathbf{v}[(n_{\mathrm{t}}, N_{\mathrm{f}})]\end{bmatrix}
\end{align}
where $\mathbf{S}\in\mathbb{C}^{N_{\mathrm{f}}\times|\mathcal{D}|K}$. In addition, assume that the lengths of prefix and suffix guard intervals satisfy the requirements in Remark 2, which implies that there is no inter-frame interference. Since the $(k+K(d-1))$-th column of $\mathbf{S}$ corresponds to the training signal of the $k$-th user associated with the $d$-th channel tap for $k\in\llbracket K\rrbracket$ and $d\in\llbracket|\mathcal{D}|\rrbracket$, which follows from \eqref{H}, the columns of $\mathbf{S}$ are required to be orthogonal for the inter-user and inter-symbol interferences to be mitigated \cite{tse2005fundamentals}. A necessary condition for $\mathbf{S}$ to have orthogonal columns in the presence of the spatial wideband effect is $N_{\mathrm{f}}\geq|\mathcal{D}|K$. In the conventional channel model, however, the requirement reduces to $N_{\mathrm{f}}\geq DK$. The frame structure and the requirements on the frame, prefix, and suffix discussed in Remarks 2 and 3 are shown in Fig. \ref{figure_3}.

\textbf{Remark 4:} Assume that the frame, prefix, and suffix satisfy the requirements in Remarks 2 and 3. To make the columns of $\mathbf{S}$ orthogonal when the spatial wideband effect is present, one possible training signal design is as follows. First, constrain the prefix and suffix training signals to be the cyclic copies of the training signals in the frame. Then, the $i$-th and $(i+Kd)$-th columns of $\mathbf{S}$ are related by $d$ circular shifts, as evident from \eqref{H}. Finally, assign a $|\mathcal{D}|k$ circularly shifted Zadoff-Chu (ZC) sequence of length $N_{\mathrm{f}}$ to the $k$-th user. Then, the columns of $\mathbf{S}$ are guaranteed to be orthogonal because circularly shifted ZC sequences are orthogonal.

\section{Proposed NFCFGS-CV Algorithm}\label{section_4}
In this section, the NFCFGS-CV-based channel estimator is proposed. In particular, the quantized received signal is partitioned to the estimation and CV data. Then, NFCFGS runs single path estimation using the estimation data, while the termination condition is checked using the CV data at each iteration. To simplify the development, the discussions on the partitioning of the data and CV-based termination condition are deferred to Section \ref{section_4c}, and NFCFGS is proposed in Sections \ref{section_4a} and \ref{section_4b}. For notational convenience, the real forms of a matrix $\mathbf{X}$ and vector $\mathbf{x}$ are denoted with the subscript $\mathrm{R}$ as
\begin{equation}\label{real_form}
\mathbf{X}_{\mathrm{R}}=\begin{bmatrix}\mathrm{Re}(\mathbf{X})&-\mathrm{Im}(\mathbf{X})\\\mathrm{Im}(\mathbf{X})&\mathrm{Re}(\mathbf{X})\end{bmatrix}\text{, }\mathbf{x}_{\mathrm{R}}=\begin{bmatrix}\mathrm{Re}(\mathbf{x})\\\mathrm{Im}(\mathbf{x})\end{bmatrix}.
\end{equation}

\subsection{Proposed NFCFGS Algorithm: Single Path Estimation}\label{section_4a}
The FCFGS algorithm \cite{doi:10.1137/090759574} is an on-grid CS algorithm that generalizes the OMP algorithm \cite{4385788} to optimization problems with sparsity-constrained convex objective functions. In this subsection, the NFCFGS algorithm is proposed for single path estimation, which is a gridless CS algorithm based on FCFGS and the NOMP algorithm \cite{7491265}. Then, NFCFGS is generalized to multipath estimation in the next subsection.

The single path estimation problem considers the case where
\begin{align}
L_{i}=\begin{cases}1&\text{if }i=k\\
                   0&\text{if }i\neq k\end{cases},
\end{align}
whose goal is to estimate the path gain, AoA, and delay of the $k$-th user, while $k\in\llbracket K\rrbracket$ is unknown in advance. To develop the MAP channel estimation framework, denote the parameters of the path as $(g, \theta, \tau, k)$ where $g$, $\theta$, and $\tau$ correspond to the path gain, AoA, and delay, whose subscript $(k, \ell)$ is suppressed for notational simplicity. In addition, define $\mathbf{a}(\theta, \tau, k)\in\mathbb{C}^{RN}$ using $(\theta, \tau, k)$, which is defined using the same logic as $\mathbf{A}(\mathcal{P})$ is introduced in \eqref{modified_unquantized_received_signal} using $\mathcal{P}$. The quantized received signal for single path estimation is written as
\begin{equation}
\hat{\mathbf{y}}=\mathrm{Q}(g\mathbf{a}(\theta, \tau, k)+\mathbf{v}).
\end{equation}

To develop the likelihood function, denote the real form of $\mathbf{a}(\theta, \tau, k)$ with the subscript $\mathrm{R}$ as\footnote{The real form of $\mathbf{a}(\theta, \tau, k)$ is constructed based on the matrix rule instead of the vector rule in \eqref{real_form}.}
\begin{equation}
\mathbf{A}_{\mathrm{R}}(\theta, \tau, k)=\begin{bmatrix}\mathrm{Re}(\mathbf{a}(\theta, \tau, k))&-\mathrm{Im}(\mathbf{a}(\theta, \tau, k))\\\mathrm{Im}(\mathbf{a}(\theta, \tau, k))&\mathrm{Re}(\mathbf{a}(\theta, \tau, k))\end{bmatrix}.
\end{equation}
In addition, the lower and upper thresholds associated with $\hat{\mathbf{y}}_{\mathrm{R}}$ are denoted as
\begin{equation}
\hat{\mathbf{y}}_{\mathrm{R}}^{\mathrm{lo}}=\begin{bmatrix}\hat{y}_{\mathrm{R}, 1}^{\mathrm{lo}}\\\vdots\\\hat{y}_{\mathrm{R}, 2RN}^{\mathrm{lo}}\end{bmatrix}\text{, }\hat{\mathbf{y}}_{\mathrm{R}}^{\mathrm{up}}=\begin{bmatrix}\hat{y}_{\mathrm{R}, 1}^{\mathrm{up}}\\\vdots\\\hat{y}_{\mathrm{R}, 2RN}^{\mathrm{up}}\end{bmatrix}
\end{equation}
as in \eqref{quantization}. Since
\begin{equation}
\mathbf{y}\sim\mathcal{CN}(g\mathbf{a}(\theta, \tau, k), \mathbf{I}_{RN})
\end{equation}
conditioned on $(g, \theta, \tau, k)$, the log-likelihood function is given as \cite{5456454}
\begin{align}\label{single_path_log_likelihood_function}
\ell_{\hat{\mathbf{y}}}(g, \theta, \tau, k)=&\log\mathrm{Pr}\left[\hat{\mathbf{y}}|g, \theta, \tau, k\right]\notag\\
                                           =&\sum_{i=1}^{2RN}\log\left(\Phi\left(\frac{\hat{y}_{\mathrm{R}, i}^{\mathrm{up}}-\mathbf{a}_{\mathrm{R}, i}^{\mathrm{T}}(\theta, \tau, k)\mathbf{g}_{\mathrm{R}}}{\sqrt{1/2}}\right)\right.\notag\\
                                            &\left.-\Phi\left(\frac{\hat{y}_{\mathrm{R}, i}^{\mathrm{lo}}-\mathbf{a}_{\mathrm{R}, i}^{\mathrm{T}}(\theta, \tau, k)\mathbf{g}_{\mathrm{R}}}{\sqrt{1/2}}\right)\right)
\end{align}
where $\Phi(\cdot)$ denotes the cumulative distribution function (CDF) of $\mathcal{N}(0, 1)$.

From \eqref{path_gain_aoa_delay} with $\ell_{\hat{\mathbf{y}}}(g, \theta, \tau, k)$, the MAP channel estimator is formulated as\footnote{The constants independent of $(g, \theta, \tau)$ in the probability density functions (PDFs) of $\mathcal{CN}(0, 1)$, $\mathrm{Uniform}([-\pi/2, \pi/2])$, and $\mathrm{Uniform}([0, (D-1)T_{s}])$ are omitted without loss of generality. In addition, the log-posterior distribution is considered.}
\begin{alignat}{2}\label{p1}
&\underset{g\in\mathbb{C}, \theta, \tau, k}{\text{maximize}}\quad&&\underbrace{\ell_{\hat{\mathbf{y}}}(g, \theta, \tau, k)-|g|^{2}}_{=f_{\hat{\mathbf{y}}}(g, \theta, \tau, k)}\notag\\
&\text{subject to}\quad                                          &&\theta\in[-\pi/2, \pi/2],\notag\\
&                                                                &&\tau\in[0, (D-1)T_{s}],\notag\\
&                                                                &&k\in\llbracket K\rrbracket.\tag{P1}
\end{alignat}
From \eqref{p1}, note that $f_{\hat{\mathbf{y}}}(g, \theta, \tau, k)$ is concave with respect to $g$ because $\ell_{\hat{\mathbf{y}}}(g, \theta, \tau, k)$ has the form of $\log(\Phi(b-g)-\Phi(a-g))$ with $b>a$, which is concave, and $-|g|^{2}$ is concave \cite{boyd2004convex}. The fact, however, that $\ell_{\hat{\mathbf{y}}}(g, \theta, \tau, k)$ is nonconvex with respect to $(\theta, \tau)$ renders \eqref{p1} highly nonconvex.

To solve \eqref{p1}, the NFCFGS algorithm is proposed based on FCFGS and NOMP, which proceeds as follows. First, $(\theta, \tau, k)$ is estimated on the grid by maximizing the gradient of $f_{\hat{\mathbf{y}}}(g, \theta, \tau, k)$. Then, $(\theta, \tau)$ is refined over the continuum using Newton's method. In path gain estimation, $g$ is estimated by maximizing $f_{\hat{\mathbf{y}}}(g, \theta, \tau, k)$ using convex optimization.

To develop NFCFGS, consider the gradient of $f_{\hat{\mathbf{y}}}(g, \theta, \tau, k)$ with respect to $g$ at $g=0$ expressed in the real form, which is
\begin{align}\label{function}
\nabla_{\mathbf{g}_{\mathrm{R}}}f_{\hat{\mathbf{y}}}(0, \theta, \tau, k)=&\begin{bmatrix}\mathrm{Re}(\nabla_{g}f_{\hat{\mathbf{y}}}(0, \theta, \tau, k))\\\mathrm{Im}(\nabla_{g}f_{\hat{\mathbf{y}}}(0, \theta, \tau, k))\end{bmatrix}\notag\\
                                                                        =&\sum_{i=1}^{2RN}\frac{\phi\left(\frac{\hat{y}_{\mathrm{R}, i}^{\mathrm{up}}}{\sqrt{1/2}}\right)-\phi\left(\frac{\hat{y}_{\mathrm{R}, i}^{\mathrm{lo}}}{\sqrt{1/2}}\right)}{\Phi\left(\frac{\hat{y}_{\mathrm{R}, i}^{\mathrm{up}}}{\sqrt{1/2}}\right)-\Phi\left(\frac{\hat{y}_{\mathrm{R}, i}^{\mathrm{lo}}}{\sqrt{1/2}}\right)}\times\notag\\
                                                                         &\left(-\frac{\mathbf{a}_{\mathrm{R}, i}(\theta, \tau, k)}{\sqrt{1/2}}\right)
\end{align}
where $\phi(\cdot)$ is the PDF of $\mathcal{N}(0, 1)$. Then, $(\theta, \tau, k)$ is estimated by solving the gradient maximization problem\footnote{$g$ is treated as a nuisance parameter by setting $g=0$, which is a standard procedure in FCFGS, and the interested reader is referred to \cite{doi:10.1137/090759574}.}
\begin{alignat}{2}\label{p2}
&\underset{\theta, \tau, k}{\text{maximize}}\quad&&\underbrace{\|\nabla_{\mathbf{g}_{\mathrm{R}}}f_{\hat{\mathbf{y}}}(0, \theta, \tau, k)\|^{2}}_{=f_{\hat{\mathbf{y}}}(\theta, \tau, k)}\notag\\
&\text{subject to}\quad                          &&\theta\in[-\pi/2, \pi/2],\notag\\
&                                                &&\tau\in[0, (D-1)T_{s}],\notag\\
&                                                &&k\in\llbracket K\rrbracket.\tag{P2}
\end{alignat}
For high-resolution ADCs, \eqref{p2} reduces to OMP because \eqref{p2} maximizes the gradient of the log-likelihood function, which is equivalent to maximizing the correlation with the residual. Before moving on, observe that $f_{\hat{\mathbf{y}}}(\theta, \tau, k)$ is nonconvex with respect to $(\theta, \tau)$.

\begin{figure*}[b]
\hrulefill
\setcounter{mytempeqncnt}{\value{equation}}
\setcounter{equation}{27}
\begin{align}
&\nabla_{(\theta, \tau)}f_{\hat{\mathbf{y}}}(\theta, \tau, k)=\left(\frac{d}{d(\theta, \tau)}\|\nabla_{\mathbf{g}_{\mathrm{R}}}f_{\hat{\mathbf{y}}}(0, \theta, \tau, k)\|^{2}\right)^{\mathrm{T}}\overset{(a)}{=}2\left(\frac{d}{d(\theta, \tau)}\nabla_{\mathbf{g}_{\mathrm{R}}}f_{\hat{\mathbf{y}}}(0, \theta, \tau, k)\right)^{\mathrm{T}}\nabla_{\mathbf{g}_{\mathrm{R}}}f_{\hat{\mathbf{y}}}(0, \theta, \tau, k),\label{gradient}\\
&\nabla_{(\theta, \tau)}^{2}f_{\hat{\mathbf{y}}}(\theta, \tau, k)\notag\\
&\overset{\hphantom{(b)}}{=}\Bigg[\begin{matrix}\frac{d}{d\theta}\bigg(2\bigg(\frac{d}{d(\theta, \tau)}\nabla_{\mathbf{g}_{\mathrm{R}}}f_{\hat{\mathbf{y}}}(0, \theta, \tau, k)\bigg)^{\mathrm{T}}\nabla_{\mathbf{g}_{\mathrm{R}}}f_{\hat{\mathbf{y}}}(0, \theta, \tau, k)\bigg)&\frac{d}{d\tau}\bigg(2\bigg(\frac{d}{d(\theta, \tau)}\nabla_{\mathbf{g}_{\mathrm{R}}}f_{\hat{\mathbf{y}}}(0, \theta, \tau, k)\bigg)^{\mathrm{T}}\nabla_{\mathbf{g}_{\mathrm{R}}}f_{\hat{\mathbf{y}}}(0, \theta, \tau, k)\bigg)\end{matrix}\Bigg]\notag\\
&\overset{(b)}{=}\Bigg[\begin{matrix}2\bigg(\bigg(\frac{d}{d\theta}\bigg(\bigg(\frac{d}{d(\theta, \tau)}\nabla_{\mathbf{g}_{\mathrm{R}}}f_{\hat{\mathbf{y}}}(0, \theta, \tau, k)\bigg)^{\mathrm{T}}\bigg)\bigg)\nabla_{\mathbf{g}_{\mathrm{R}}}f_{\hat{\mathbf{y}}}(0, \theta, \tau, k)+\bigg(\frac{d}{d(\theta, \tau)}\nabla_{\mathbf{g}_{\mathrm{R}}}f_{\hat{\mathbf{y}}}(0, \theta, \tau, k)\bigg)^{\mathrm{T}}\frac{d}{d\theta}\nabla_{\mathbf{g}_{\mathrm{R}}}f_{\hat{\mathbf{y}}}(0, \theta, \tau, k)\bigg)\end{matrix}\notag\\
&\mathrel{\hphantom{\overset{\hphantom{(b)}}{=}}}\begin{matrix}2\bigg(\bigg(\frac{d}{d\tau}\bigg(\bigg(\frac{d}{d(\theta, \tau)}\nabla_{\mathbf{g}_{\mathrm{R}}}f_{\hat{\mathbf{y}}}(0, \theta, \tau, k)\bigg)^{\mathrm{T}}\bigg)\bigg)\nabla_{\mathbf{g}_{\mathrm{R}}}f_{\hat{\mathbf{y}}}(0, \theta, \tau, k)+\bigg(\frac{d}{d(\theta, \tau)}\nabla_{\mathbf{g}_{\mathrm{R}}}f_{\hat{\mathbf{y}}}(0, \theta, \tau, k)\bigg)^{\mathrm{T}}\frac{d}{d\tau}\nabla_{\mathbf{g}_{\mathrm{R}}}f_{\hat{\mathbf{y}}}(0, \theta, \tau, k)\bigg)\end{matrix}\Bigg]\label{hessian}
\end{align}
\setcounter{equation}{\value{mytempeqncnt}}
\end{figure*}

To solve \eqref{p2}, NFCFGS proceeds as follows. First, $(\theta, \tau, k)$ is estimated on the grid
\begin{equation}
\Omega\subseteq[-\pi/2, \pi/2]\times[0, (D-1)T_{s}]\times\llbracket K\rrbracket,
\end{equation}
which discretizes the constraints in \eqref{p2}. The sizes of the grids that correspond to $(\theta, \tau)$ are denoted as $G_{\mathrm{a}}=R_{\mathrm{a}}M$ and $G_{\mathrm{d}}=R_{\mathrm{d}}|\mathcal{D}|$ where $R_{\mathrm{a}}$ and $R_{\mathrm{d}}$ are the grid resolutions. In practice, $\Omega$ is configured so that the constraints in \eqref{p2} are uniformly discretized, while the grid resolutions are constrained to satisfy $R_{\mathrm{a}}\geq 1$ and $R_{\mathrm{d}}\geq 1$ \cite{7491265, 7458188, 7961152}. Then, $(\theta, \tau)$ is refined using Newton's method to reduce the off-grid error.

\setcounter{equation}{29}

To apply Newton's method, the Newton step is required, so the gradient and Hessian of $f_{\hat{\mathbf{y}}}(\theta, \tau, k)$ are derived, which are denoted as $\nabla_{(\theta, \tau)}f_{\hat{\mathbf{y}}}(\theta, \tau, k)\in\mathbb{R}^{2}$ and $\nabla_{(\theta, \tau)}^{2}f_{\hat{\mathbf{y}}}(\theta, \tau, k)\in\mathbb{R}^{2\times 2}$. The gradient is formulated as \eqref{gradient} where (a) comes from the chain rule. The Hessian is derived as \eqref{hessian}, which is obtained by applying the product rule in (b). The derivatives of $\nabla_{\mathbf{g}_{\mathrm{R}}}f_{\hat{\mathbf{y}}}(0, \theta, \tau, k)$ that are required to evaluate the gradient and Hessian are provided as
\begin{align}\label{derivative}
\frac{d^{i+j}}{d\theta^{i}d\tau^{j}}\nabla_{\mathbf{g}_{\mathrm{R}}}f_{\hat{\mathbf{y}}}(0, \theta, \tau, k)=&\sum_{n=1}^{2RN}\frac{\phi\left(\frac{\hat{y}_{\mathrm{R}, n}^{\mathrm{up}}}{\sqrt{1/2}}\right)-\phi\left(\frac{\hat{y}_{\mathrm{R}, n}^{\mathrm{lo}}}{\sqrt{1/2}}\right)}{\Phi\left(\frac{\hat{y}_{\mathrm{R}, n}^{\mathrm{up}}}{\sqrt{1/2}}\right)-\Phi\left(\frac{\hat{y}_{\mathrm{R}, n}^{\mathrm{lo}}}{\sqrt{1/2}}\right)}\times\notag\\
                                                                                                             &\frac{d^{i+j}}{d\theta^{i}d\tau^{j}}\left(-\frac{\mathbf{a}_{\mathrm{R}, n}(\theta, \tau, k)}{\sqrt{1/2}}\right).
\end{align}
Then, the gradient step and Newton step are given as \cite{boyd2004convex}
\begin{equation}
\begin{aligned}
\mathbf{g}_{\hat{\mathbf{y}}}(\theta, \tau, k)&=\nabla_{(\theta, \tau)}f_{\hat{\mathbf{y}}}(\theta, \tau, k),\\
\mathbf{n}_{\hat{\mathbf{y}}}(\theta, \tau, k)&=-\nabla_{(\theta, \tau)}^{2}f_{\hat{\mathbf{y}}}(\theta, \tau, k)^{-1}\nabla_{(\theta, \tau)}f_{\hat{\mathbf{y}}}(\theta, \tau, k),
\end{aligned}
\end{equation}
which are evaluated based on \eqref{function}, \eqref{gradient}, \eqref{hessian}, and \eqref{derivative}. From the discussions until now, NFCFGS is summarized as follows.

\textbf{AoA-delay estimation:} $(\theta, \tau, k)$ is estimated on the grid by solving \eqref{p2} as
\begin{equation}\label{on_grid_estimation}
(\hat{\theta}, \hat{\tau}, \hat{k})=\argmax_{(\theta, \tau, k)\in\Omega}f_{\hat{\mathbf{y}}}(\theta, \tau, k).
\end{equation}
Then, $(\hat{\theta}, \hat{\tau})$ is refined based on Newton's method for nonconvex functions as \cite{murphy2012machine}
\begin{align}\label{off_grid_estimation}
(\hat{\theta}, \hat{\tau})+\begin{cases}\eta\mathbf{n}_{\hat{\mathbf{y}}}(\hat{\theta}, \hat{\tau}, \hat{k})&\text{if }\nabla_{(\theta, \tau)}^{2}f_{\hat{\mathbf{y}}}(\hat{\theta}, \hat{\tau}, \hat{k})\prec\mathbf{0}_{2\times 2}\\
                                        \eta\mathbf{g}_{\hat{\mathbf{y}}}(\hat{\theta}, \hat{\tau}, \hat{k})&\text{if }\nabla_{(\theta, \tau)}^{2}f_{\hat{\mathbf{y}}}(\hat{\theta}, \hat{\tau}, \hat{k})\nprec\mathbf{0}_{2\times 2}\end{cases}
\end{align}
at each iteration where $\eta$ is the step size. The purpose of the condition in \eqref{off_grid_estimation} is to allow the Newton refinement only when $f_{\hat{\mathbf{y}}}(\theta, \tau, k)$ is locally concave, which is commonly adopted when Newton's method is applied to nonconvex functions \cite{murphy2012machine, 10.1007/BFb0067700, 7491265}.

\textbf{Path gain estimation:} $g$ is estimated by solving \eqref{p1} with fixed $(\hat{\theta}, \hat{\tau}, \hat{k})$ as
\begin{equation}
\hat{g}=\argmax_{g\in\mathbb{C}}f_{\hat{\mathbf{y}}}(g, \hat{\theta}, \hat{\tau}, \hat{k}),
\end{equation}
which is guaranteed to converge to the global optimum due to the concavity of $f_{\hat{\mathbf{y}}}(g, \theta, \tau, k)$ with respect to $g$.

\subsection{Proposed NFCFGS Algorithm: Multipath Estimation}\label{section_4b}
This subsection proposes NFCFGS for multipath estimation where single path estimation is performed at each iteration. To illustrate how NFCFGS proceeds at each iteration, define $\hat{\mathcal{P}}$ as the collection of all previously estimated $(\theta_{k, \ell}, \tau_{k, \ell})$, and $\hat{\bm{\alpha}}\in\mathbb{C}^{|\hat{\mathcal{P}}|}$ as the collection of all previously estimated $\alpha_{k, \ell}$. The goal of each iteration is to estimate the path from the $k$-th user using $(\hat{\bm{\alpha}}, \hat{\mathcal{P}})$, while $k\in\llbracket K\rrbracket$ is unknown a priori. For the sake of simplicity, assume that $(\hat{\bm{\alpha}}, \hat{\mathcal{P}})$ contains no error.

To extend single path estimation in the previous subsection to multipath estimation, denote the parameters of the path of interest as $(g, \theta, \tau, k)$. In addition, define $\mathbf{a}(\theta, \tau, k)\in\mathbb{C}^{RN}$ using $(\theta, \tau, k)$ from the same logic in the previous subsection. The quantized received signal is written as
\begin{equation}
\hat{\mathbf{y}}=\mathrm{Q}(g\mathbf{a}(\theta, \tau, k)+\underbrace{\mathbf{A}(\hat{\mathcal{P}})\hat{\bm{\alpha}}}_{=\mathbf{y}(\hat{\bm{\alpha}}, \hat{\mathcal{P}})}+\mathbf{v}),
\end{equation}
while the unquantized received signal is distributed as
\begin{equation}\label{conditional_distribution}
\mathbf{y}\sim\mathcal{CN}(g\mathbf{a}(\theta, \tau, k)+\mathbf{y}(\hat{\bm{\alpha}}, \hat{\mathcal{P}}), \mathbf{I}_{RN})
\end{equation}
conditioned on $(g, \theta, \tau, k)$ and $(\hat{\bm{\alpha}}, \hat{\mathcal{P}})$. Then, the log-likelihood function is derived from \eqref{conditional_distribution} as
\begin{align}\label{multipath_log_likelihood_function}
&\log\mathrm{Pr}\left[\hat{\mathbf{y}}|g, \theta, \tau, k, \hat{\bm{\alpha}}, \hat{\mathcal{P}}\right]\notag\\
&=\sum_{i=1}^{2RN}\log\left(\Phi\left(\frac{\hat{y}_{\mathrm{R}, i}^{\mathrm{up}}-y_{\mathrm{R}, i}(\hat{\bm{\alpha}}, \hat{\mathcal{P}})-\mathbf{a}_{\mathrm{R}, i}^{\mathrm{T}}(\theta, \tau, k)\mathbf{g}_{\mathrm{R}}}{\sqrt{1/2}}\right)\right.\notag\\
&\mathrel{\hphantom{=}}\left.-\Phi\left(\frac{\hat{y}_{\mathrm{R}, i}^{\mathrm{lo}}-y_{\mathrm{R}, i}(\hat{\bm{\alpha}}, \hat{\mathcal{P}})-\mathbf{a}_{\mathrm{R}, i}^{\mathrm{T}}(\theta, \tau, k)\mathbf{g}_{\mathrm{R}}}{\sqrt{1/2}}\right)\right),
\end{align}
or
\begin{equation}\label{identity}
\log\mathrm{Pr}\left[\hat{\mathbf{y}}|g, \theta, \tau, k, \hat{\bm{\alpha}}, \hat{\mathcal{P}}\right]=\ell_{\hat{\mathbf{y}}-\mathbf{y}(\hat{\bm{\alpha}}, \hat{\mathcal{P}})}(g, \theta, \tau, k),
\end{equation}
which can be verified by comparing \eqref{multipath_log_likelihood_function} with \eqref{single_path_log_likelihood_function}.

From the discussions until now, we illustrate how NFCFGS estimates $(\theta, \tau, k)$, while $(\hat{\bm{\alpha}}, \hat{\mathcal{P}})$ is assumed to be fixed. As in the previous subsection, consider the log-posterior distribution of $(g, \theta, \tau)$, which is derived from \eqref{identity}. The square of the magnitude of the gradient with respect to $g$ is formulated as $f_{\hat{\mathbf{y}}-\mathbf{y}(\hat{\bm{\alpha}}, \hat{\mathcal{P}})}(\theta, \tau, k)$. Therefore, $(\theta, \tau, k)$ is estimated from \eqref{p2} after $\hat{\mathbf{y}}$ is replaced with $\hat{\mathbf{y}}-\mathbf{y}(\hat{\bm{\alpha}}, \hat{\mathcal{P}})$. In particular, $(\theta, \tau, k)$ is estimated by applying \eqref{on_grid_estimation} and \eqref{off_grid_estimation} after $\hat{\mathbf{y}}$ is replaced with $\hat{\mathbf{y}}-\mathbf{y}(\hat{\bm{\alpha}}, \hat{\mathcal{P}})$. Then, NFCFGS updates $\hat{\mathcal{P}}$ as $\hat{\mathcal{P}}\cup\{(\hat{\theta}, \hat{\tau}, \hat{k})\}$.

Now, we show how the path gains of $\hat{\mathcal{P}}$ are estimated, whose collected path gains are denoted as $\mathbf{x}\in\mathbb{C}^{|\hat{\mathcal{P}}|}$. In other words, the path gains of the currently and previously estimated paths are updated. Assume that
\begin{equation}
\hat{\mathbf{y}}=\mathrm{Q}(\mathbf{A}(\hat{\mathcal{P}})\mathbf{x}+\mathbf{v})
\end{equation}
with $\mathbf{x}\sim\mathcal{CN}(\mathbf{0}_{|\mathcal{P}|}, \mathbf{I}_{|\mathcal{P}|})$, while $\hat{\mathcal{P}}$ is assumed to be fixed. The log-likelihood function is given as \cite{5456454}
\begin{align}
\ell(\mathbf{x}, \hat{\mathcal{P}})=&\log\mathrm{Pr}\left[\hat{\mathbf{y}}|\mathbf{x}, \hat{\mathcal{P}}\right]\notag\\
                                   =&\sum_{i=1}^{2RN}\log\left(\Phi\left(\frac{\hat{y}_{\mathrm{R}, i}^{\mathrm{up}}-\mathbf{a}_{\mathrm{R}, i}^{\mathrm{T}}(\hat{\mathcal{P}})\mathbf{x}_{\mathrm{R}}}{\sqrt{1/2}}\right)\right.\notag\\
                                    &\left.-\Phi\left(\frac{\hat{y}_{\mathrm{R}, i}^{\mathrm{lo}}-\mathbf{a}_{\mathrm{R}, i}^{\mathrm{T}}(\hat{\mathcal{P}})\mathbf{x}_{\mathrm{R}}}{\sqrt{1/2}}\right)\right).
\end{align}
Then, the MAP estimate of $\mathbf{x}$ is
\begin{equation}
\hat{\bm{\alpha}}=\argmax_{\mathbf{x}\in\mathbb{C}^{|\hat{\mathcal{P}}|}}(\underbrace{\ell(\mathbf{x}, \hat{\mathcal{P}})-\|\mathbf{x}\|^{2}}_{=g(\mathbf{x}, \hat{\mathcal{P}})}),
\end{equation}
whose global optimum is guaranteed due to the concavity of $g(\mathbf{x}, \hat{\mathcal{P}})$ with respect to $\mathbf{x}$.

The NFCFGS algorithm is presented in Algorithm \ref{nfcfgs_cv}, whose CV-based termination condition is explained in the next subsection, as well as the meaning of $\mathcal{E}$ and $\mathcal{CV}$ in the superscripts and subscripts. In particular, Line 5 performs on-grid AoA-delay estimation. Then, $(\hat{\theta}, \hat{\tau})$ is refined with Newton's method provided that the objective function is locally concave in Line 8, but is refined using gradient descent method otherwise in Line 10. After $\hat{\mathcal{P}}$ is updated in Line 13, path gain estimation is performed using convex optimization in Line 14.

\begin{algorithm}[t]
\caption{NFCFGS-CV algorithm}\label{nfcfgs_cv}
\begin{algorithmic}[1]
\Require $f_{\hat{\mathbf{y}}}^{\mathcal{E}}(\cdot, \cdot, \cdot)$, $\mathbf{g}_{\hat{\mathbf{y}}}^{\mathcal{E}}(\cdot, \cdot, \cdot)$, $\mathbf{n}_{\hat{\mathbf{y}}}^{\mathcal{E}}(\cdot, \cdot, \cdot)$, $g_{\mathcal{E}}(\cdot, \cdot)$, $\ell_{\mathcal{CV}}(\cdot, \cdot)$
\Ensure $(\hat{\bm{\alpha}}, \hat{\mathcal{P}})$
\State // $\ell_{\mathcal{CV}}(\hat{\bm{\alpha}}, \emptyset)=-\infty$ and $\mathbf{y}(\hat{\bm{\alpha}}, \emptyset)=\mathbf{0}_{2RN}$ by convention
\State $\hat{\mathcal{P}}\coloneqq\emptyset$
\Do
\State $\epsilon\coloneqq\ell_{\mathcal{CV}}(\hat{\bm{\alpha}}, \hat{\mathcal{P}})$
\State $(\hat{\theta}, \hat{\tau}, \hat{k})\coloneqq\displaystyle\argmax_{(\theta, \tau, k)\in\Omega}f_{\hat{\mathbf{y}}-\mathbf{y}(\hat{\bm{\alpha}}, \hat{\mathcal{P}})}^{\mathcal{E}}(\theta, \tau, k)$
\While {termination condition}
\If {$\nabla_{(\theta, \tau)}^{2}f_{\hat{\mathbf{y}}-\mathbf{y}(\hat{\bm{\alpha}}, \hat{\mathcal{P}})}^{\mathcal{E}}(\hat{\theta}, \hat{\tau}, \hat{k})\prec\mathbf{0}_{2\times 2}$}
\State $(\hat{\theta}, \hat{\tau})\coloneqq(\hat{\theta}, \hat{\tau})+\eta\mathbf{n}_{\hat{\mathbf{y}}-\mathbf{y}(\hat{\bm{\alpha}}, \hat{\mathcal{P}})}^{\mathcal{E}}(\hat{\theta}, \hat{\tau}, \hat{k})$
\Else
\State $(\hat{\theta}, \hat{\tau})\coloneqq(\hat{\theta}, \hat{\tau})+\eta\mathbf{g}_{\hat{\mathbf{y}}-\mathbf{y}(\hat{\bm{\alpha}}, \hat{\mathcal{P}})}^{\mathcal{E}}(\hat{\theta}, \hat{\tau}, \hat{k})$
\EndIf
\EndWhile
\State $\hat{\mathcal{P}}\coloneqq\hat{\mathcal{P}}\cup\{(\hat{\theta}, \hat{\tau}, \hat{k})\}$
\State $\hat{\bm{\alpha}}\coloneqq\displaystyle\argmax_{\mathbf{x}\in\mathbb{C}^{|\hat{\mathcal{P}}|}}g_{\mathcal{E}}(\mathbf{x}, \hat{\mathcal{P}})$
\doWhile $\ell_{\mathcal{CV}}(\hat{\bm{\alpha}}, \hat{\mathcal{P}})>\epsilon$
\end{algorithmic}
\end{algorithm}

\subsection{Proposed CV Technique: Termination Condition}\label{section_4c}
The knowledge on $L$ is important to determine the proper termination condition of NFCFGS. In practice, however, $L$ is difficult to acquire, which demands a more feasible termination condition. In \cite{4301267, 5319752}, the termination conditions of CS algorithms were determined based on the CV technique \cite{hastie2009elements}, which is a model validation technique that assesses the quality of the estimate to prevent overfitting. In this subsection, the CV-based termination condition is applied to NFCFGS, which gives the NFCFGS-CV-based channel estimator.

To explain the concept of NFCFGS-CV, consider the estimation data $\hat{\mathbf{y}}_{\mathcal{E}}\in\mathbb{C}^{|\mathcal{E}|}$ and CV data $\hat{\mathbf{y}}_{\mathcal{CV}}\in\mathbb{C}^{|\mathcal{CV}|}$ that partition $\hat{\mathbf{y}}$, which means that $\mathcal{E}$ and $\mathcal{CV}$ partition $\llbracket RN\rrbracket$. Then, channel estimation is performed based on the estimation data, while the estimation quality is assessed based on the CV data. The disjoint nature of the estimation and CV data enables CV to properly assess the estimation quality. For convenience, define $\mathcal{E}_{\mathrm{R}}$ and $\mathcal{CV}_{\mathrm{R}}$, which denote the real index sets of the estimation and CV data, as
\begin{equation}
\begin{aligned}
 \mathcal{E}_{\mathrm{R}}&=\{i\in\mathbb{Z}\mid i\in\mathcal{E}\cup(\mathcal{E}+RN)\},\\
\mathcal{CV}_{\mathrm{R}}&=\{i\in\mathbb{Z}\mid i\in\mathcal{CV}\cup(\mathcal{CV}+RN)\}
\end{aligned}
\end{equation}
where the addition of a scalar to a set represents an operation that adds the scalar to each element of the set.

To illustrate the details of NFCFGS-CV, define the estimation functions as
\begin{equation}\label{estimation_function}
\begin{alignedat}{2}
&f_{\hat{\mathbf{y}}}^{\mathcal{E}}(\theta, \tau, k)         &&\text{: square of the magnitude of the gradient},\\
&\mathbf{g}_{\hat{\mathbf{y}}}^{\mathcal{E}}(\theta, \tau, k)&&\text{: gradient step},\\
&\mathbf{n}_{\hat{\mathbf{y}}}^{\mathcal{E}}(\theta, \tau, k)&&\text{: Newton step},\\
&g_{\mathcal{E}}(\mathbf{x}, \hat{\mathcal{P}})              &&\text{: log-posterior distribution},
\end{alignedat}
\end{equation}
and the CV function as
\begin{equation}\label{cv_function}
\ell_{\mathcal{CV}}(\mathbf{x}, \hat{\mathcal{P}})\text{: log-likelihood function}
\end{equation}
where the estimation (CV) functions are obtained by retaining only the terms corresponding to $\mathcal{E}_{\mathrm{R}}$ ($\mathcal{CV}_{\mathrm{R}}$) in the summations of $f_{\hat{\mathbf{y}}}(\theta, \tau)$, $\mathbf{g}_{\hat{\mathbf{y}}}(\theta, \tau, k)$, $\mathbf{n}_{\hat{\mathbf{y}}}(\theta, \tau, k)$, $g(\mathbf{x}, \hat{\mathcal{P}})$, and $\ell(\mathbf{x}, \hat{\mathcal{P}})$. Then, the NFCFGS-CV algorithm proceeds as follows, whose details are presented in Algorithm \ref{nfcfgs_cv}. First, NFCFGS performs channel estimation in Lines 5, 8, 10, 13, and 14 as discussed in the previous subsections, but uses \eqref{estimation_function} instead of $f_{\hat{\mathbf{y}}}(\theta, \tau)$, $\mathbf{g}_{\hat{\mathbf{y}}}(\theta, \tau, k)$, $\mathbf{n}_{\hat{\mathbf{y}}}(\theta, \tau, k)$, and $g(\mathbf{x}, \hat{\mathcal{P}})$. Then, Line 15 checks the CV-based termination condition using \eqref{cv_function}, which terminates NFCFGS when the estimation quality starts to decrease due to overfitting.

To demonstrate how NFCFGS-CV proceeds, consider a toy example where Algorithm \ref{nfcfgs_cv} terminates in three iterations with $\hat{\mathcal{P}}$ updated in Line 13 as
\begin{align}
           &\{(\hat{\theta}_{2, 1}, \hat{\tau}_{2, 1}, 2)\}\notag\\
\rightarrow&\{(\hat{\theta}_{2, 1}, \hat{\tau}_{2, 1}, 2), (\hat{\theta}_{1, 1}, \hat{\tau}_{1, 1}, 1)\}\notag\\
\rightarrow&\{(\hat{\theta}_{2, 1}, \hat{\tau}_{2, 1}, 2), (\hat{\theta}_{1, 1}, \hat{\tau}_{1, 1}, 1), (\hat{\theta}_{2, 2}, \hat{\tau}_{2, 2}, 2)\},\notag
\end{align}
and the quality of $\hat{\mathcal{P}}$ is assessed in Line 15. In this example, the CV-based termination condition is satisfied in the third iteration, which implies that running Algorithm \ref{nfcfgs_cv} beyond three iterations leads to overfitting. Therefore, NFCFGS-CV is terminated, and the estimated number of paths is $\hat{L}_{1}=1$ and $\hat{L}_{2}=2$. Again, note that no prior knowledge on $L_{1}$ and $L_{2}$ was assumed, and the estimated number of paths is obtained from CV.

To show how CV detects overfitting, which comes from the disjoint nature of the estimation and CV data, consider the SE, which is defined as
\begin{equation}
\mathrm{SE}=\|\hat{\mathbf{h}}-\mathbf{h}\|^{2}
\end{equation}
where $\hat{\mathbf{h}}=\mathbf{F}(\hat{\mathcal{P}})\hat{\bm{\alpha}}$. An illustration of how $\ell_{\mathcal{CV}}(\hat{\bm{\alpha}}, \hat{\mathcal{P}})$ and the SE evolves with the iteration of NFCFGS-CV at $\mathrm{SNR}=0$ dB with $4$-bit ADCs is shown in Fig. \ref{figure_4}. To emphasize the disjoint nature of the estimation and CV data, $\ell_{\mathcal{E}}(\hat{\bm{\alpha}}, \hat{\mathcal{P}})$ is provided as a reference, which is defined as in \eqref{cv_function}. In Fig. \ref{figure_4}, the decreasing point of $\ell_{\mathcal{CV}}(\hat{\bm{\alpha}}, \hat{\mathcal{P}})$ agrees with the minimum SE, which clearly indicates overfitting. In contrast, however, $\ell_{\mathcal{E}}(\hat{\bm{\alpha}}, \hat{\mathcal{P}})$ increases monotonically.

\begin{figure}[t]
\centering
\includegraphics[width=1\columnwidth]{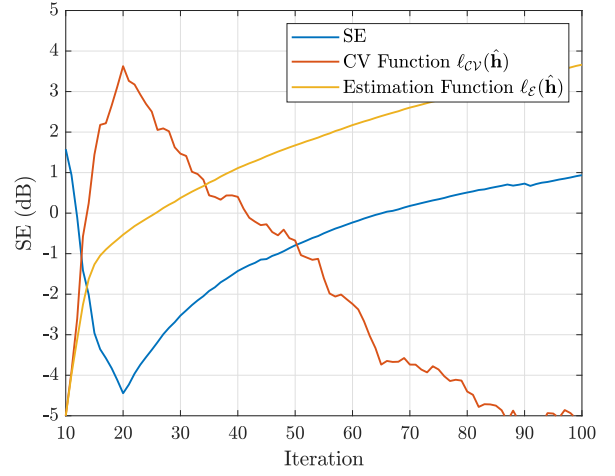}
\caption[caption]{SE in dB scale versus iteration of NFCFGS-CV with vertically scaled and shifted $\ell_{\mathcal{CV}}(\hat{\bm{\alpha}}, \hat{\mathcal{P}})$ and $\ell_{\mathcal{E}}(\hat{\bm{\alpha}}, \hat{\mathcal{P}})$ in linear scale for $B=4$ and $\mathrm{SNR}=0$ dB.}\label{figure_4}
\end{figure}

To investigate how CV assesses the estimation quality, the relationship between $\ell_{\mathcal{CV}}(\hat{\bm{\alpha}}, \hat{\mathcal{P}})$ and the SE is analyzed based on the CS perspective, which interprets $\bar{\mathbf{A}}$ in \eqref{unquantized_received_signal} as a sensing matrix with deterministic $\mathbf{h}$ and $\hat{\mathbf{h}}$. To facilitate the analysis, assume that $\bar{\mathbf{A}}$ has i.i.d. elements with zero mean and unit variance, which is a common assumption in the CS literature \cite{eldar2012compressed, 1614066}.\footnote{This assumption makes $\bar{\mathbf{A}}$ a favorable sensing matrix because rectangular sensing matrices with i.i.d. elements demonstrate well-conditionedness as the matrix size increases \cite{bai1993limit, eldar2012compressed}. In practice, however, how to design the training signals and RF combiners to make $\bar{\mathbf{A}}$ well-conditioned remains unclear, and is left as future work.} Then, an alternative expression of $\ell_{\mathcal{CV}}(\hat{\bm{\alpha}}, \hat{\mathcal{P}})$ is given as
\begin{align}
\ell_{\mathcal{CV}}(\hat{\mathbf{h}})=&\ell_{\mathcal{CV}}(\hat{\bm{\alpha}}, \hat{\mathcal{P}})\notag\\
                                     =&\sum_{i\in\mathcal{CV}_{\mathrm{R}}}\log\left(\Phi\left(\frac{\hat{y}_{\mathrm{R}, i}^{\mathrm{up}}-\bar{\mathbf{a}}_{\mathrm{R}, i}^{\mathrm{T}}\hat{\mathbf{h}}_{\mathrm{R}}}{\sqrt{1/2}}\right)\right.\notag\\
                                      &\left.-\Phi\left(\frac{\hat{y}_{\mathrm{R}, i}^{\mathrm{lo}}-\bar{\mathbf{a}}_{\mathrm{R}, i}^{\mathrm{T}}\hat{\mathbf{h}}_{\mathrm{R}}}{\sqrt{1/2}}\right)\right)
\end{align}
to emphasize the dependence on $\hat{\mathbf{h}}$. The following lemma shows how $\ell_{\mathcal{CV}}(\hat{\mathbf{h}})$ behaves in the asymptotic regime of $|\mathcal{CV}|$.
\begin{lemma}\label{lemma}
Assume that the elements of $\bar{\mathbf{A}}$ are i.i.d. with zero mean and unit variance. Define the limit of $1/2|\mathcal{CV}|\cdot\ell_{\mathcal{CV}}(\hat{\mathbf{h}})$ as $f_{\mathcal{CV}}(\hat{\mathbf{h}})$. In particular, $1/2|\mathcal{CV}|\cdot\ell_{\mathcal{CV}}(\hat{\mathbf{h}})\to f_{\mathcal{CV}}(\hat{\mathbf{h}})$ as $|\mathcal{CV}|\to\infty$ in probability. Then, $f_{\mathcal{CV}}(\hat{\mathbf{h}})$ is a concave function of $\hat{\mathbf{h}}$, whose maximum is at $\hat{\mathbf{h}}=\mathbf{h}$.
\end{lemma}
\begin{proof}
The randomness of $\ell_{\mathcal{CV}}(\hat{\mathbf{h}})$ comes from $\hat{\mathbf{y}}$ and $\bar{\mathbf{A}}$. Since $\hat{\mathbf{y}}=\mathrm{Q}(\bar{\mathbf{A}}\mathbf{h}+\mathbf{v})$, the elements of $\hat{\mathbf{y}}$ are i.i.d. as are the elements of $\bar{\mathbf{A}}$. Then, the weak law of large numbers gives
\begin{align}\label{sample_mean}
&f_{\mathcal{CV}}(\hat{\mathbf{h}})\notag\\
&\overset{\hphantom{(a)}}{=}\mathbb{E}_{\hat{y}_{\mathrm{R}, i}, \bar{\mathbf{a}}_{\mathrm{R}, i}}\left\{\log\left(\Phi\left(\frac{\hat{y}_{\mathrm{R}, i}^{\mathrm{up}}-\bar{\mathbf{a}}_{\mathrm{R}, i}^{\mathrm{T}}\hat{\mathbf{h}}_{\mathrm{R}}}{\sqrt{1/2}}\right)\right.\right.\notag\\
&\mathrel{\hphantom{\overset{\hphantom{(a)}}{=}}}\left.\left.-\Phi\left(\frac{\hat{y}_{\mathrm{R}, i}^{\mathrm{lo}}-\bar{\mathbf{a}}_{\mathrm{R}, i}^{\mathrm{T}}\hat{\mathbf{h}}_{\mathrm{R}}}{\sqrt{1/2}}\right)\right)\right\}\notag\\
&\overset{(a)}{=}\sum_{\hat{y}_{\mathrm{R}, i}\in\mathcal{Q}}\mathbb{E}_{\bar{\mathbf{a}}_{\mathrm{R}, i}\mid\hat{y}_{\mathrm{R}, i}}\left\{\log\left(\Phi\left(\frac{\hat{y}_{\mathrm{R}, i}^{\mathrm{up}}-\bar{\mathbf{a}}_{\mathrm{R}, i}^{\mathrm{T}}\hat{\mathbf{h}}_{\mathrm{R}}}{\sqrt{1/2}}\right)\right.\right.\notag\\
&\mathrel{\hphantom{\overset{\hphantom{(a)}}{=}}}\left.\left.-\Phi\left(\frac{\hat{y}_{\mathrm{R}, i}^{\mathrm{lo}}-\bar{\mathbf{a}}_{\mathrm{R}, i}^{\mathrm{T}}\hat{\mathbf{h}}_{\mathrm{R}}}{\sqrt{1/2}}\right)\right)\right\}\mathrm{Pr}\left[\hat{y}_{\mathrm{R}, i}\right]\notag\\
&\overset{(b)}{=}\sum_{\hat{y}_{\mathrm{R}, i}\in\mathcal{Q}}\mathbb{E}_{\bar{\mathbf{a}}_{\mathrm{R}, i}}\left\{\log\left(\Phi\left(\frac{\hat{y}_{\mathrm{R}, i}^{\mathrm{up}}-\bar{\mathbf{a}}_{\mathrm{R}, i}^{\mathrm{T}}\hat{\mathbf{h}}_{\mathrm{R}}}{\sqrt{1/2}}\right)\right.\right.\notag\\
&\mathrel{\hphantom{\overset{\hphantom{(b)}}{=}}}\left.\left.-\Phi\left(\frac{\hat{y}_{\mathrm{R}, i}^{\mathrm{lo}}-\bar{\mathbf{a}}_{\mathrm{R}, i}^{\mathrm{T}}\hat{\mathbf{h}}_{\mathrm{R}}}{\sqrt{1/2}}\right)\right)\right.\left(\Phi\left(\frac{\hat{y}_{\mathrm{R}, i}^{\mathrm{up}}-\bar{\mathbf{a}}_{\mathrm{R}, i}^{\mathrm{T}}\mathbf{h}_{\mathrm{R}}}{\sqrt{1/2}}\right)\right.\notag\\
&\mathrel{\hphantom{\overset{\hphantom{(b)}}{=}}}\left.\left.-\Phi\left(\frac{\hat{y}_{\mathrm{R}, i}^{\mathrm{lo}}-\bar{\mathbf{a}}_{\mathrm{R}, i}^{\mathrm{T}}\mathbf{h}_{\mathrm{R}}}{\sqrt{1/2}}\right)\right)\right\}\notag\\
&\overset{\hphantom{(b)}}{=}\mathbb{E}_{\bar{\mathbf{a}}_{\mathrm{R}, i}}\left\{\sum_{\hat{y}_{\mathrm{R}, i}\in\mathcal{Q}}\log\left(\Phi\left(\frac{\hat{y}_{\mathrm{R}, i}^{\mathrm{up}}-\bar{\mathbf{a}}_{\mathrm{R}, i}^{\mathrm{T}}\hat{\mathbf{h}}_{\mathrm{R}}}{\sqrt{1/2}}\right)\right.\right.\notag\\
&\mathrel{\hphantom{\overset{\hphantom{(b)}}{=}}}\left.\left.-\Phi\left(\frac{\hat{y}_{\mathrm{R}, i}^{\mathrm{lo}}-\bar{\mathbf{a}}_{\mathrm{R}, i}^{\mathrm{T}}\hat{\mathbf{h}}_{\mathrm{R}}}{\sqrt{1/2}}\right)\right)\right.\left(\Phi\left(\frac{\hat{y}_{\mathrm{R}, i}^{\mathrm{up}}-\bar{\mathbf{a}}_{\mathrm{R}, i}^{\mathrm{T}}\mathbf{h}_{\mathrm{R}}}{\sqrt{1/2}}\right)\right.\notag\\
&\mathrel{\hphantom{\overset{\hphantom{(b)}}{=}}}\left.\left.-\Phi\left(\frac{\hat{y}_{\mathrm{R}, i}^{\mathrm{lo}}-\bar{\mathbf{a}}_{\mathrm{R}, i}^{\mathrm{T}}\mathbf{h}_{\mathrm{R}}}{\sqrt{1/2}}\right)\right)\right\}
\end{align}
where (a) is the result of conditioning on $\hat{y}_{\mathrm{R}, i}$, (b) comes from the fact that
\begin{align}
&\mathsf{p}_{\bar{\mathbf{a}}_{\mathrm{R}, i}\mid\hat{y}_{\mathrm{R}, i}}(\bar{\mathbf{a}}_{\mathrm{R}, i})\mathrm{Pr}\left[\hat{y}_{\mathrm{R}, i}\right]\notag\\
&\overset{\hphantom{(c)}}{=}\mathrm{Pr}\left[\hat{y}_{\mathrm{R}, i}\mid\bar{\mathbf{a}}_{\mathrm{R}, i}\right]\mathsf{p}_{\bar{\mathbf{a}}_{\mathrm{R}, i}}(\bar{\mathbf{a}}_{\mathrm{R}, i})\notag\\
&\overset{(c)}{=}\left(\Phi\left(\frac{\hat{y}_{\mathrm{R}, i}^{\mathrm{up}}-\bar{\mathbf{a}}_{\mathrm{R}, i}^{\mathrm{T}}\mathbf{h}_{\mathrm{R}}}{\sqrt{1/2}}\right)-\Phi\left(\frac{\hat{y}_{\mathrm{R}, i}^{\mathrm{lo}}-\bar{\mathbf{a}}_{\mathrm{R}, i}^{\mathrm{T}}\mathbf{h}_{\mathrm{R}}}{\sqrt{1/2}}\right)\right)\times\notag\\
&\mathrel{\hphantom{\overset{\hphantom{(c)}}{=}}}\mathsf{p}_{\bar{\mathbf{a}}_{\mathrm{R}, i}}(\bar{\mathbf{a}}_{\mathrm{R}, i})
\end{align}
where $\mathsf{p}(\cdot)$ is the probability measure on $\bar{\mathbf{a}}_{\mathrm{R}, i}$, and (c) follows from the same logic in \eqref{single_path_log_likelihood_function}. Since expectation is a nonnegative weighted sum, many properties inside the expectation remain unchanged. Therefore, we focus on the summation inside the expectation in the last line of \eqref{sample_mean}. First, note that the summation inside the expectation is concave with respect to $\hat{\mathbf{h}}$ because each term has the form of $\log(\Phi(b-\mathbf{c}^{\mathrm{T}}\hat{\mathbf{h}})-\Phi(a-\mathbf{c}^{\mathrm{T}}\hat{\mathbf{h}}))$ scaled by a nonnegative number with $b>a$, which is concave \cite{boyd2004convex}. Therefore, $f_{\mathcal{CV}}(\hat{\mathbf{h}})$ is concave with respect to $\hat{\mathbf{h}}$. In addition, note that the summation inside the expectation is the negative cross entropy of the probability mass functions (PMFs), which are obtained by quantizing $\mathcal{N}(\bar{\mathbf{a}}_{\mathrm{R}, i}^{\mathrm{T}}\mathbf{h}_{\mathrm{R}}, 1/2)$ and $\mathcal{N}(\bar{\mathbf{a}}_{\mathrm{R}, i}^{\mathrm{T}}\hat{\mathbf{h}}_{\mathrm{R}}, 1/2)$ using $\mathrm{Q}(\cdot)$. Since the maximum of the negative cross entropy occurs when the PMFs are identical, the maximum of $f_{\mathcal{CV}}(\hat{\mathbf{h}})$ is attained at $\hat{\mathbf{h}}=\mathbf{h}$, which completes the proof.
\end{proof}
Lemma \ref{lemma} says that $\ell_{\mathcal{CV}}(\hat{\mathbf{h}})$ is a concave function of $\hat{\mathbf{h}}$, whose maximum is attained at $\hat{\mathbf{h}}=\mathbf{h}$, provided that the number of the CV data is sufficiently large. Therefore, $\ell_{\mathcal{CV}}(\hat{\mathbf{h}})$ partially reflects the SE. The overall SE, however, is clouded by the distortion that is suspected to be incurred by the nonlinearity of low-resolution ADCs. To analyze how the distortion decreases with the ADC resolution, consider $\infty$-bit ADCs with uniform quantization intervals of width $\Delta$. Then, the following corollary shows how the distortion is related to $\Delta$, whose reciprocal is interpreted as the ADC resolution.
\begin{corollary}\label{corollary}
Assume that $\infty$-bit ADCs with uniform quantization intervals of width $\Delta$ are used. Then, $f_{\mathcal{CV}}(\hat{\mathbf{h}})$ is expressed as
\begin{equation}
f_{\mathcal{CV}}(\hat{\mathbf{h}})=-\|\hat{\mathbf{h}}-\mathbf{h}\|^{2}-\frac{1}{2}\log\pi e+\log\Delta+\mathcal{O}(\Delta).
\end{equation}
\end{corollary}
\begin{proof}
Consider the discrete random variable $X^{\Delta}$, which is obtained by quantizing the continuous random variable $X$ with $\infty$-bit ADCs with uniform quantization intervals of width $\Delta$. Then, the entropy $H(X^{\Delta})$ is related to the differential entropy $h(X)$ as \cite{cover2012elements}
\begin{equation}\label{relationship}
H(X^{\Delta})+\log\Delta=h(X)+\mathcal{O}(\Delta).
\end{equation}
Since the summation inside the expectation in the last line of \eqref{sample_mean} is the negative cross entropy that is obtained by quantizing $\mathcal{N}(\bar{\mathbf{a}}_{\mathrm{R}, i}^{\mathrm{T}}\mathbf{h}_{\mathrm{R}}, 1/2)$ and $\mathcal{N}(\bar{\mathbf{a}}_{\mathrm{R}, i}^{\mathrm{T}}\hat{\mathbf{h}}_{\mathrm{R}}, 1/2)$, whose differential cross entropy is \cite{cover2012elements}
\begin{equation}\label{differental_cross_entropy}
\left(\bar{\mathbf{a}}_{\mathrm{R}, i}^{\mathrm{T}}\left(\hat{\mathbf{h}}_{\mathrm{R}}-\mathbf{h}_{\mathrm{R}}\right)\right)^{2}+\frac{1}{2}\log\pi e,
\end{equation}
applying \eqref{relationship} and \eqref{differental_cross_entropy} to the summation inside the expectation in the last line of \eqref{sample_mean} gives
\begin{align}
&f_{\mathcal{CV}}(\hat{\mathbf{h}})-\log\Delta\notag\\
&\overset{\hphantom{(a)}}{=}\mathbb{E}_{\bar{\mathbf{a}}_{\mathrm{R}, i}}\left\{-\left(\bar{\mathbf{a}}_{\mathrm{R}, i}^{\mathrm{T}}\left(\hat{\mathbf{h}}_{\mathrm{R}}-\mathbf{h}_{\mathrm{R}}\right)\right)^{2}-\frac{1}{2}\log\pi e+\mathcal{O}(\Delta)\right\}\notag\\
&\overset{(a)}{=}-\|\hat{\mathbf{h}}_{\mathrm{R}}-\mathbf{h}_{\mathrm{R}}\|^{2}-\frac{1}{2}\log\pi e+\mathcal{O}(\Delta)
\end{align}
where (a) is due to the fact that the elements of $\bar{\mathbf{a}}_{\mathrm{R}, i}$ are i.i.d. with zero mean and unit variance.
\end{proof}
Corollary \ref{corollary} says that $\ell_{\mathcal{CV}}(\hat{\mathbf{h}})$ is partially capable of assessing the SE when the number of the CV data is sufficiently large, and the accuracy increases with the ADC resolution. In particular, the terms that are not proportional to the SE decreases as $\mathcal{O}(\Delta)$. To empirically verify Corollary \ref{corollary} for the general case where low-resolution ADCs are employed, the contour plot of $f_{\mathcal{CV}}(\hat{\mathbf{h}})$ is illustrated in Fig. \ref{figure_5} for various $B$, while the elements of $\bar{\mathbf{A}}$ are i.i.d. as $\mathcal{CN}(0, 1)$, and the expectation in $f_{\mathcal{CV}}(\hat{\mathbf{h}})$ is numerically evaluated. According to Fig. \ref{figure_5}, $f_{\mathcal{CV}}(\hat{\mathbf{h}})$ becomes more proportional to the SE as $B$ increases, which agrees with Corollary \ref{corollary}.

\begin{figure}[t]
\centering
\includegraphics[width=1\columnwidth]{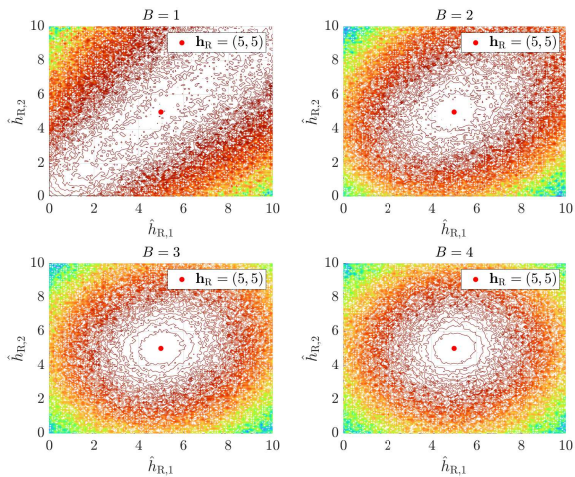}
\caption[caption]{Contour plot of $f_{\mathcal{CV}}(\hat{\mathbf{h}})$ for $B=1, 2, 3, 4$ with $\mathbf{h}_{\mathrm{R}}=(5, 5)$. The elements of $\bar{\mathbf{A}}$ are i.i.d. as $\mathcal{CN}(0, 1)$.}\label{figure_5}
\end{figure}

\section{Simulation Results}\label{section_5}
In this section, the performance of the proposed NFCFGS-CV-based channel estimator is evaluated based on the normalized mean SE (NMSE), which is defined as
\begin{equation}
\mathrm{NMSE}=\mathbb{E}\left\{\frac{\|\hat{\mathbf{h}}-\mathbf{h}\|^{2}}{\|\mathbf{h}\|^{2}}\right\},
\end{equation}
along with the complexity analysis. The base station employs a ULA with half-wavelength antenna spacings of $\mathsf{d}=\lambda_{c}/2$, while $M$ and $R$ vary from simulation to simulation. The columns of the RF combiner are configured as circularly shifted ZC sequences of length $M$. A pair of $B$-bit ADCs with uniform quantization intervals are deployed at each RF chain, whose quantization intervals vary with the SNR. The training signals are selected as circularly shifted ZC sequences of length $N_{\mathrm{f}}$ with the cyclic prefix and suffix as discussed in Remark 4, while $N_{\mathrm{f}}$, $N_{\mathrm{p}}$, and $N_{\mathrm{s}}$ satisfy the requirements in Remarks 2 and 3. The SNRs of all the users are assumed to be the same unless otherwise stated.\footnote{In practical multiuser systems, the users whose received powers are below a certain threshold are excluded to avoid those users from being clouded by other users during the sampling process, which is commonly achieved by properly scheduling the users along with power control.} The raised-cosine (RC) pulse shaping filter is selected with a roll-off factor of $0.35$. The system operates in the carrier frequency of $f_{c}=28$ GHz and bandwidth $W=600$ MHz, while the channel parameters $D$ and $L_{k}$ vary from simulation to simulation. Throughout the simulations, the spatial wideband effect is incorporated by generating the channel according to \eqref{spatial_wideband_effect}, whose path gains, AoAs, and delays are assumed to be i.i.d. as \eqref{path_gain_aoa_delay} (except in Fig. \ref{figure_6} where the AoA is deterministically generated). The grid resolutions of $\Omega$ for on-grid AoA-delay estimation are set as $(R_{\mathrm{a}}, R_{\mathrm{d}})=(2, 2)$. Among $N_{\mathrm{t}}$ frames, $0.8N_{\mathrm{t}}$ frames are used for estimation, while $0.2N_{\mathrm{t}}$ frames are devoted to CV for NFCFGS-CV. In contrast, other benchmarks utilize all $N_{\mathrm{t}}$ frames for estimation.

To illustrate the impact of the spatial wideband effect on the channel estimation performance, the NMSEs of the two channel estimators are compared, one that accounts for the propagation delay across the antenna array as proposed and one that neglects the spatial wideband effect. The latter performs channel estimation based on the conventional channel model that assumes $p(dT_{s}-\tau_{k, \ell, m})\approx p(dT_{s}-\tau_{k, \ell})$ in \eqref{spatial_wideband_effect}, so the performance degradation due to the model mismatch is inevitable. Since the propagation delay across the antenna array becomes severe as the AoA deviates from the boresight, namely $\theta=0$, the performance degradation is expected to increase with the AoA deviation. To measure the performance degradation due to the model mismatch, define the NMSE degradation as
\begin{equation}
\mathrm{NMSE}\text{ }\mathrm{degradation}=\frac{\mathbb{E}\{\|\check{\mathbf{h}}-\mathbf{h}\|^{2}/\|\mathbf{h}\|^{2}\}}{\mathbb{E}\{\|\hat{\mathbf{h}}-\mathbf{h}\|^{2}/\|\mathbf{h}\|^{2}\}}
\end{equation}
where $\check{\mathbf{h}}$ is the channel estimate that does not account for the spatial wideband effect, while $\hat{\mathbf{h}}$ does.

In Fig. \ref{figure_6}, the NMSE degradation of NFCFGS-CV is shown in the single-user single-path case for various AoAs with $B=\infty$, $M=256$, $R=16$, $K=1$, $D=2$, $L=1$, $N_{\mathrm{t}}=80$, and $N_{\mathrm{f}}=10$. According to Fig. \ref{figure_6}, the NMSE degradation due to the model mismatch increases as the AoA deviation from the boresight increases, which suggests that the spatial wideband effect cannot be neglected in wideband mmWave massive MIMO systems. Therefore, we conclude that properly modeling the propagation delay across the antenna array has a significant impact on the channel estimation performance.

\begin{figure}[t]
\centering
\includegraphics[width=1\columnwidth]{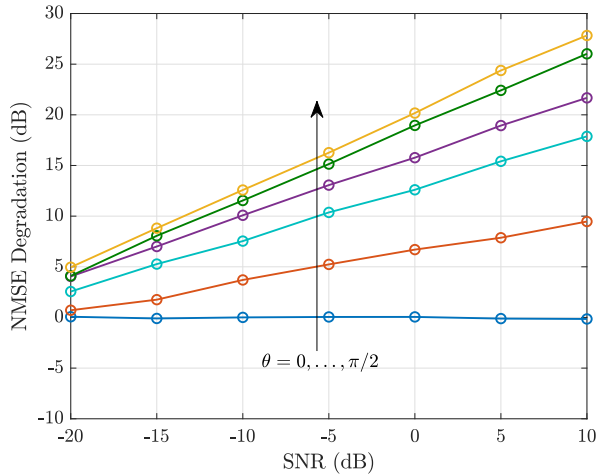}
\caption[caption]{NMSE degradation due to the model mismatch for $\theta=0, \dots, \pi/2$ versus SNR for $B=\infty$ with $M=256$, $R=16$, $K=1$, $D=2$, $L=1$, $N_{\mathrm{t}}=80$, and $N_{\mathrm{f}}=10$.}\label{figure_6}
\end{figure}

To evaluate the performance of NFCFGS-CV, the GAMP \cite{8171203}, VAMP \cite{8171203}, GEC-SR \cite{8320852}, and VB-SBL \cite{8310593} algorithms are considered as benchmarks, which are state-of-the-art on-grid CS algorithms for generalized linear models. The grid resolutions of $\Omega$ for the construction of sensing matrices are selected as $(R_{\mathrm{a}}, R_{\mathrm{d}})=(2, 2)$ for VB-SBL. GAMP, VAMP, and GEC-SR, however, are constrained to $(R_{\mathrm{a}}, R_{\mathrm{d}})=(1, 1)$ because these algorithms are sensitive to ill-conditioned sensing matrices. GAMP, VAMP, and GEC-SR use the Gaussian mixture model to approximate the prior distribution of $\mathbf{h}$, whereas VB-SBL adopts the Student-t distribution \cite{8310593}. In addition, the performance of the FCFGS-CV algorithm without Newton's method \cite{8789658} is shown as a reference to demonstrate how the gridless approach improves the performance.

\begin{figure}[t]
\centering
\includegraphics[width=1\columnwidth]{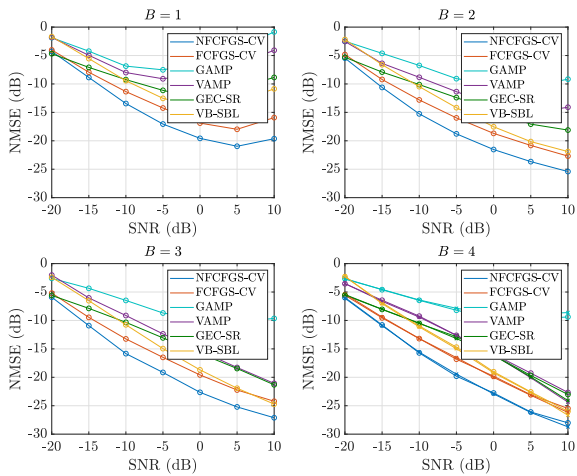}
\caption[caption]{NMSE versus SNR for $B=1, 2, 3, 4$ with $M=32$, $R=8$, $K=4$, $D=4$, $(L_{1}, L_{2}, L_{3}, L_{4})=(4, 2, 4, 2)$, $N_{\mathrm{t}}=40$, and $N_{\mathrm{f}}=40$. In the last subplot, $B=\infty$ is shown as a reference with a cross.}\label{figure_7}
\end{figure}

\begin{figure}[t]
\centering
\includegraphics[width=1\columnwidth]{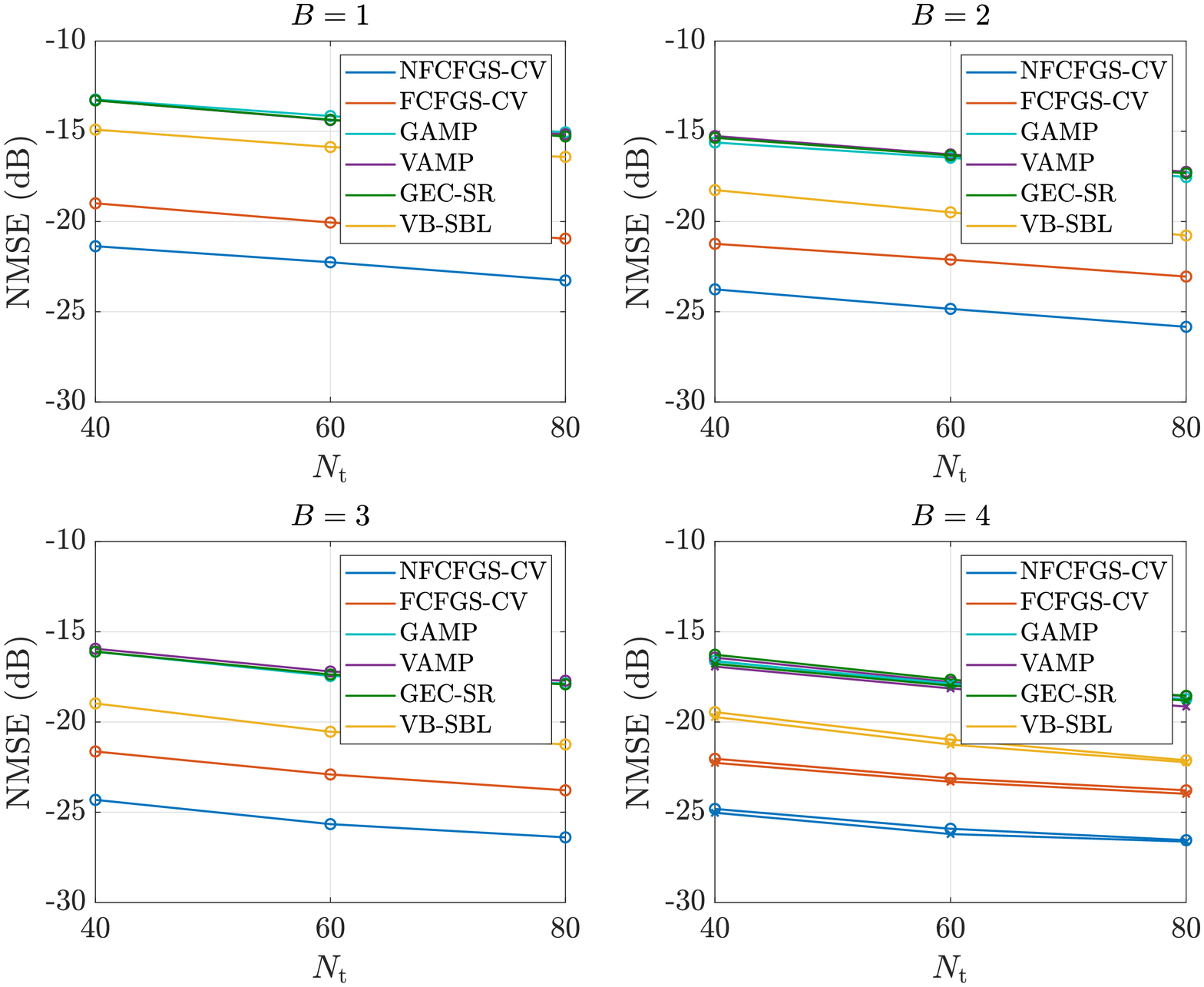}
\caption[caption]{NMSE versus $N_{\mathrm{t}}$ for $B=1, 2, 3, 4$ with $M=32$, $R=8$, $K=4$, $D=4$, $L_{k}=2$ for all $k$, $N_{\mathrm{f}}=40$, and $\mathrm{SNR}=0$ dB. In the last subplot, $B=\infty$ is shown as a reference with a cross.}\label{figure_8}
\end{figure}

\begin{figure}[t]
\centering
\includegraphics[width=1\columnwidth]{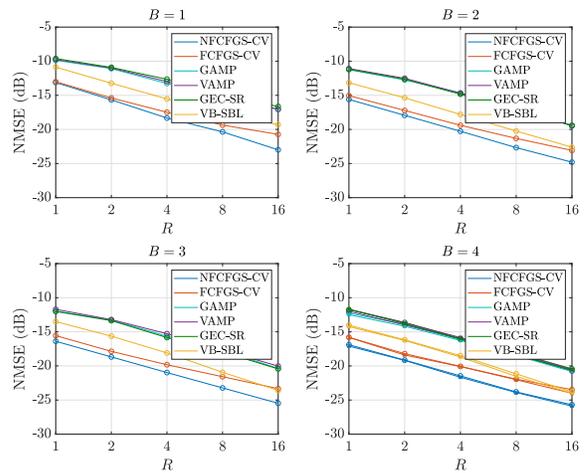}
\caption[caption]{NMSE versus $R$ for $B=1, 2, 3, 4$ with $M=16$, $K=2$, $D=4$, $L_{k}=2$ for all $k$, $N=1600$, and $\mathrm{SNR}=0$ dB. In the last subplot, $B=\infty$ is shown as a reference with a cross.}\label{figure_9}
\end{figure}

\begin{figure}[t]
\centering
\includegraphics[width=1\columnwidth]{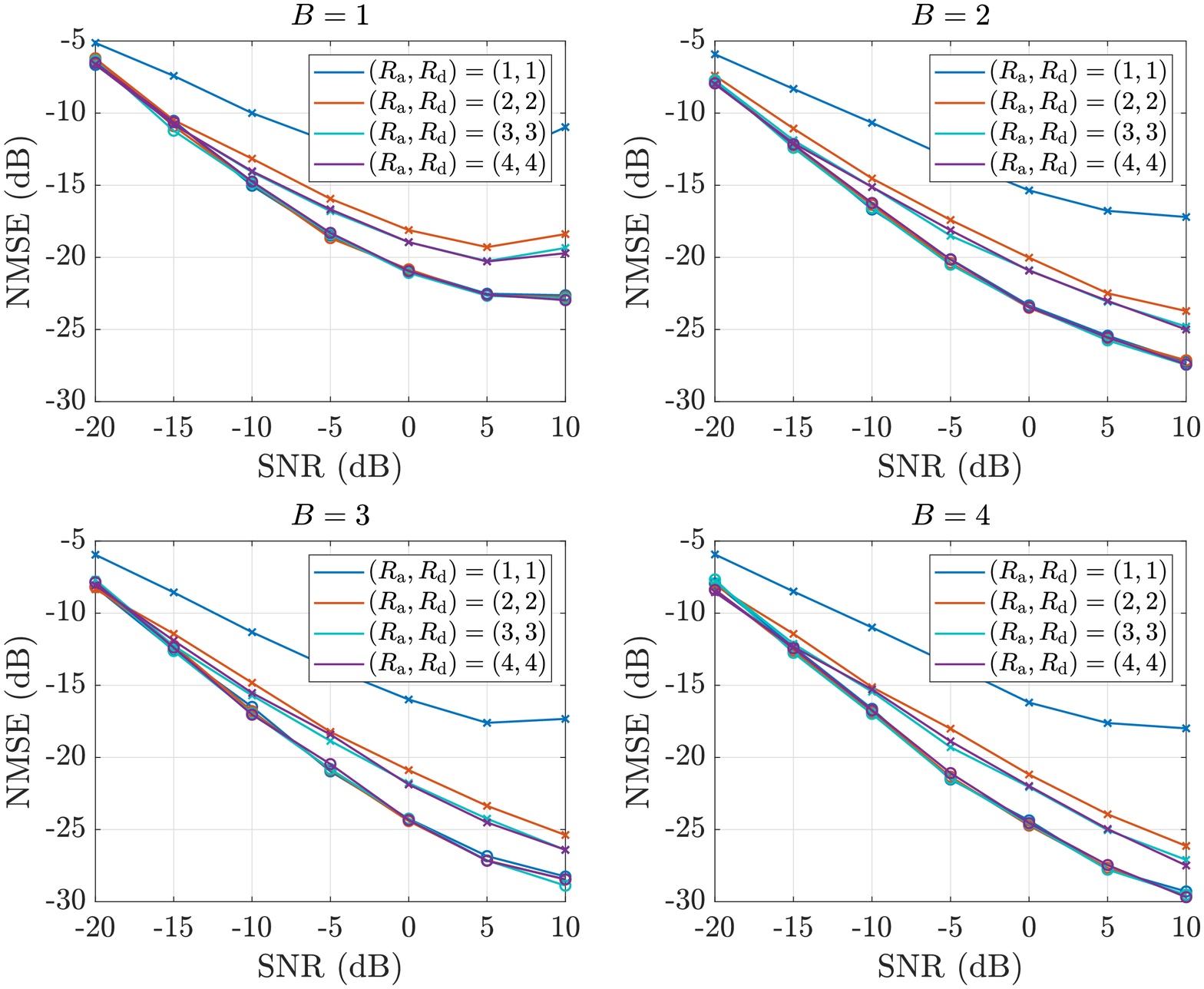}
\caption[caption]{NMSE for $(R_{\mathrm{a}}, R_{\mathrm{d}})=(1, 1), (2, 2), (3, 3), (4, 4)$ versus SNR for $B=1, 2, 3, 4$ with $M=32$, $R=8$, $K=4$, $L_{k}=2$ for all $k$, $N_{\mathrm{t}}=40$, and $N_{\mathrm{f}}=40$. In each subplot, circles and crosses represent NFCFGS-CV and FCFGS-CV.}\label{figure_10}
\end{figure}

\begin{table}[t]
\centering
\caption{Per-iteration complexities of various channel estimators}\label{table_1}
\renewcommand{\arraystretch}{1.2}
\begin{tabular}{c||c}
\hline
Algorithm&Complexity\\
\hline\hline
NFCFGS-CV&$\mathcal{O}(MRN|\mathcal{D}|KR_{\mathrm{a}}R_{\mathrm{d}}+RNkt_{\mathrm{G}}+RNt_{\mathrm{N}})$\\
\hline
FCFGS-CV&$\mathcal{O}(MRN|\mathcal{D}|KR_{\mathrm{a}}R_{\mathrm{d}}+RNkt_{\mathrm{G}})$ \cite{8789658}\\
\hline
GAMP&$\mathcal{O}(MRN|\mathcal{D}|KR_{\mathrm{a}}R_{\mathrm{d}})$ \cite{8171203}\\
\hline
VAMP&$\mathcal{O}(MRN|\mathcal{D}|KR_{\mathrm{a}}R_{\mathrm{d}})$ \cite{8171203}\\
\hline
GEC-SR&$\mathcal{O}(M^{2}RN|\mathcal{D}|^{2}K^{2}R_{\mathrm{a}}^{2}R_{\mathrm{d}}^{2})$ \cite{8320852}\\
\hline
VB-SBL&$\mathcal{O}(M^{2}RN|\mathcal{D}|^{2}K^{2}R_{\mathrm{a}}^{2}R_{\mathrm{d}}^{2})$ \cite{8310593}\\
\hline
\end{tabular}
\end{table}

\begin{table}[t]
\centering
\caption{Average numbers of iterations of (NFCFGS-CV, FCFGS-CV) for various SNRs and $B$ when $(R_{\mathrm{a}}, R_{\mathrm{d}})=(2, 2)$}\label{table_2}
\renewcommand{\arraystretch}{1.2}
\begin{tabular}{c||c||c||c||c}
\hline
SNR (dB)&$B=1$&$B=2$&$B=3$&$B=4$\\
\hline\hline
$-20$&$(11, 12)$&$(12, 13)$&$(12, 13)$&$(12, 12)$\\
\hline
$-15$&$(13, 14)$&$(13, 16)$&$(13, 16)$&$(13, 16)$\\
\hline
$-10$&$(14, 18)$&$(15, 20)$&$(15, 20)$&$(15, 21)$\\
\hline
$-5$&$(16, 22)$&$(17, 25)$&$(19, 26)$&$(18, 27)$\\
\hline
$0$&$(18, 26)$&$(21, 33)$&$(22, 35)$&$(22, 36)$\\
\hline
$5$&$(20, 33)$&$(24, 40)$&$(26, 47)$&$(27, 51)$\\
\hline
$10$&$(22, 36)$&$(26, 47)$&$(30, 57)$&$(31, 65)$\\
\hline
\end{tabular}
\end{table}

In Fig. \ref{figure_7}, the NMSEs of NFCFGS-CV and other benchmarks are provided for various SNRs with $B=1, 2, 3, 4$, while the other parameters are configured as $M=32$, $R=8$, $K=4$, $D=4$, $(L_{1}, L_{2}, L_{3}, L_{4})=(4, 2, 4, 2)$, $N_{\mathrm{t}}=40$, and $N_{\mathrm{f}}=40$. In this simulation, users with different SNRs are considered, namely $\rho_{k}/\rho_{1}=2(k-1)$ dB with $\mathrm{SNR}=1/K\cdot(\rho_{1}+\cdots+\rho_{K})$. As shown in Fig. \ref{figure_7}, NFCFGS-CV outperforms other benchmarks for all SNRs and $B$. The inferior performance of GAMP, VAMP, GEC-SR, and VB-SBL comes from the fact that these algorithms attempt to estimate the continuous parameters of $\mathbf{h}$ on the grid, which leads to the off-grid error. NFCFGS-CV, on the other hand, adopts the gridless approach to mitigate the off-grid error, thereby improving the performance. In addition, CV enables NFCFGS-CV to terminate when the minimum SE is achieved as discussed in Lemma \ref{lemma} and Corollary \ref{corollary}. The NMSEs of all algorithms, however, increase in the high SNR regime especially when $B=1$ because the magnitude information is eliminated by the coarse quantization. To alleviate the performance degradation in the high SNR regime, the quantization thresholds are configured based on the SNR when $B=2, 3, 4$, which is commonly adopted when low-resolution ADCs are used \cite{8171203, 8320852, 8310593}. An in-depth comparison between NFCFGS-CV and FCFGS-CV is conducted in the simulation in Fig. \ref{figure_10}.

In Fig. \ref{figure_8}, the NMSEs of various channel estimators are shown for different $N_{\mathrm{t}}$ with $B=1, 2, 3, 4$, and the other parameters are given as $M=32$, $R=8$, $K=4$, $D=4$, $L_{k}=2$ for all $k$, $N_{\mathrm{f}}=40$, and $\mathrm{SNR}=0$ dB. Similar to Fig. \ref{figure_7}, NFCFGS-CV is superior to other benchmarks at all $N_{\mathrm{t}}$ and $B$. In addition, note that NFCFGS-CV performs satisfactorily even when $N_{\mathrm{t}}$ is not large, which suggests that NFCFGS-CV provides low training overhead.

In Fig. \ref{figure_9}, the NMSEs of NFCFGS-CV and other benchmarks are compared at various $R$ with $B=1, 2, 3, 4$, while the other parameters are selected as $M=16$, $K=2$, $D=4$, $L_{k}=2$ for all $k$, $N=1600$, and $\mathrm{SNR}=0$ dB. In contrast to the previous simulations, $N_{\mathrm{t}}$ and $N_{\mathrm{f}}$ are configured based on $R$ to maximize the channel estimation performance. As evident in Fig. \ref{figure_9}, NFCFGS-CV outperforms all algorithms even in the extreme case of $R=1$ where only a unidirectional RF combiner is available at each frame, which reconfirms the superior performance of NFCFGS-CV.

Now, the performances and complexities of NFCFGS-CV and FCFGS-CV are compared. First, recall that NFCFGS is gridless FCFGS that refines the estimated channel over the continuum using Newton's method. Therefore, NFCFGS can be regarded as FCFGS over the infinite-resolution grid. In Fig. \ref{figure_10}, the NMSEs of NFCFGS-CV and FCFGS-CV are compared at various grid resolutions with $B=1, 2, 3, 4$, while the other parameters are set as $M=32$, $R=8$, $K=4$, $D=4$, $L_{k}=2$ for all $k$, $N_{\mathrm{t}}=40$, and $N_{\mathrm{f}}=40$. According to Fig. \ref{figure_10}, $(R_{\mathrm{a}}, R_{\mathrm{d}})=(1, 1)$ is sufficient for NFCFGS-CV because Newton's method effectively eliminates the off-grid error even when the grid is coarse. FCFGS-CV, on the other hand, requires the grid resolution to be high to offset the off-grid error for a satisfactory performance. The performance, however, saturates at $(R_{\mathrm{a}}, R_{\mathrm{d}})=(4, 4)$ because the off-grid error cannot be eliminated solely with on-grid AoA-delay estimation. This suggests that FCFGS-CV cannot approach NFCFGS-CV even when the grid resolution is high.

To compare the complexities of NFCFGS-CV and FCFGS-CV, denote the numbers of iterations for Newton's method and gradient descent method in Lines 8 and 14 of Algorithm \ref{nfcfgs_cv} to converge as $t_{\mathrm{G}}$ and $t_{\mathrm{N}}$, which are typically less than $10$. Then, the $k$-th iteration complexities of NFCFGS-CV and FCFGS-CV are shown in Table \ref{table_1}. In addition, Table \ref{table_2} shows the average numbers of iterations for NFCFGS-CV and FCFGS-CV to terminate for various SNRs and $B$ when $(R_{\mathrm{a}}, R_{\mathrm{d}})=(2, 2)$, while the other parameters are configured as in Fig. \ref{figure_10}. According to Table \ref{table_1}, the per-iteration complexities of NFCFGS-CV and FCFGS-CV are dominated by $MRN|\mathcal{D}|KR_{\mathrm{a}}R_{\mathrm{d}}$ and $RNkt_{\mathrm{G}}$ where the first term is mainly determined by the grid resolution. Therefore, the complexities of NFCFGS-CV and FCFGS-CV strongly depend on the iteration number and grid resolution, namely $R_{\mathrm{a}}$ and $R_{\mathrm{d}}$. FCFGS-CV, however, requires the grid resolution to be high to offset the off-grid error, which increases the per-iteration complexity. Moreover, Table \ref{table_2} suggests that FCFGS-CV requires more iterations to terminate because FCFGS-CV attempts to reduce the off-grid error by overestimating the number of paths, or the model order in the CS jargon, which is a phenomenon that on-grid CS algorithms typically demonstrate \cite{7491265, 6834753}. Therefore, we conclude that NFCFGS-CV outperforms FCFGS-CV in terms of both performance and complexity, as evident from Fig. \ref{figure_10} and Tables \ref{table_1} and \ref{table_2}.

Also, the complexities of NFCFGS-CV and other benchmarks are compared. According to Table \ref{table_1}, the complexities of NFCFGS-CV, GAMP, and VAMP have the same order, whereas those of GEC-SR and VB-SBL increase by a factor of $M|\mathcal{D}|KR_{\mathrm{a}}R_{\mathrm{d}}$. Therefore, NFCFGS-CV, GAMP, and VAMP are suitable for massive MIMO systems where $M$ is typically large, whereas GEC-SR and VB-SBL are infeasible.

\begin{figure}[t]
\centering
\includegraphics[width=1\columnwidth]{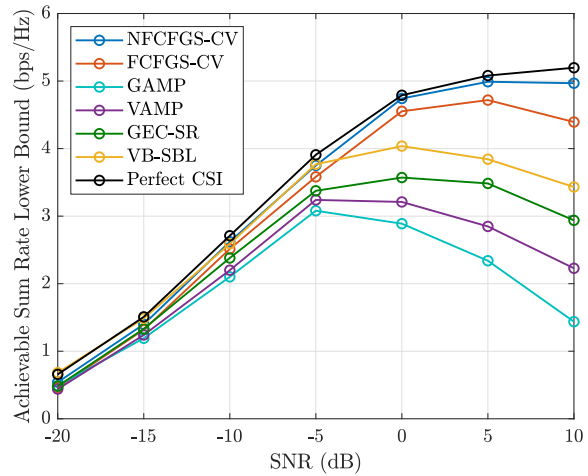}
\caption[caption]{Achievable sum rate lower bound \cite{7931630} versus SNR for $B=1$ with $M=32$, $R=8$, $K=4$, $D=4$, $(L_{1}, L_{2}, L_{3}, L_{4})=(4, 2, 4, 2)$, $N_{\mathrm{t}}=40$, $N_{\mathrm{f}}=40$, and $N_{\mathrm{c}}=9000$.}\label{figure_11}
\end{figure}

Now, the achievable sum rate lower bounds of various channel estimators are investigated based on the framework in \cite{7931630}. In this simulation, the coherence block of length $N_{\mathrm{c}}$ is partitioned to the channel estimation phase of length $N_{\mathrm{t}}(N_{\mathrm{p}}+N_{\mathrm{f}}+N_{\mathrm{s}})$ and uplink data transmission phase. In the uplink data transmission phase, the data is transmitted in blocks with prefix and suffix guard intervals to avoid inter-block interference. The base station and users employ the block diagonalization-based combiners and precoders \cite{1261332}, which are obtained from the estimated channel. For one-bit ADCs, the achievable sum rate lower bound $(N_{\mathrm{c}}-N_{\mathrm{t}}(N_{\mathrm{p}}+N_{\mathrm{f}}+N_{\mathrm{s}}))/N_{\mathrm{c}}\cdot R_{\mathrm{sum}}$ can be evaluated based on \cite{7931630}. In essence, the Bussgang theorem is applied to linearize one-bit ADCs, whose covariance matrix of the quantization noise is given by the arcsine law. Then, the achievable sum rate lower bound is obtained using the worst noise assumption, which treats the quantization noise as Gaussian with the same covariance matrix. The details are omitted for conciseness, and the interested reader is referred to \cite{7931630}.

In Fig. \ref{figure_11}, the achievable sum rate lower bounds of various channel estimators and the perfect channel state information (CSI) case are shown for various SNRs with $B=1$ and $N_{\mathrm{c}}=9000$, while the other parameters are configured as in Fig. \ref{figure_7}. According to Fig. \ref{figure_11} in conjunction with the first subplot in Fig. \ref{figure_7}, NFCFGS-CV is again superior to other benchmarks for all SNRs, which shows that lower NMSE in the channel estimation phase guarantees higher reliability in the data transmission phase. As discussed earlier, the degradation in the high SNR regime is due to the coarse quantization of one-bit ADCs.

\section{Conclusion}\label{section_6}
The NFCFGS-CV-based channel estimator was proposed for wideband mmWave massive MIMO systems with hybrid architectures and low-resolution ADCs. The channel was formulated based on the spatial wideband effect in the discrete time domain. As a means to mitigate inter-frame, inter-user, and inter-symbol interferences in the presence of the spatial wideband effect, the requirements on the training signal design were investigated. To estimate the channel, the MAP criterion was adopted, which was solved using the NFCFGS algorithm. In addition, the CV technique was applied to deal with the unknown number of paths. The analysis on CV revealed that the minimum SE can be detected by CV when the number of the CV data is sufficiently large. The simulation results showed that NFCFGS-CV is accurate and efficient, which outperformed state-of-the-art channel estimators.

\bibliographystyle{IEEEtran}
\bibliography{refs_all}

\end{document}